\documentclass[aps,pra,showpacs,amsmath,amssymb,floatfix,superscriptaddress,notitlepage]{revtex4-1}

\usepackage{color}
\usepackage{bbm}
\usepackage{graphicx}
\usepackage{dcolumn}
\usepackage[utf8]{inputenc}

\usepackage{graphicx}
\usepackage{epstopdf}
\usepackage{amsthm}
\usepackage{amsmath}
\usepackage{empheq}
\usepackage{bbm}
\usepackage{braket}
\usepackage{amssymb}
\usepackage[usenames,dvipsnames]{xcolor}
\usepackage{cases}
\usepackage{latexsym}
\usepackage[colorlinks=true,citecolor=Cerulean,linkcolor=RubineRed,urlcolor=Cerulean]{hyperref}
\usepackage[title]{appendix}

\graphicspath{ {images/}{images/figure1/}{images/figure2/}{images/figure3/}{images/figure4/}{images/figure5/} }

\newtheorem{theorem}{Theorem}

\newtheorem{definition}{Definition}

\renewcommand{\v}[1]{\ensuremath{\boldsymbol #1}}

\DeclareMathOperator{\tr}{Tr}

\usepackage{lipsum}
\usepackage{graphicx}

\begin{document}
\raggedbottom	
	\title{No-resonance conditions, random matrices, and quantum chaotic models}
	
\author{Jonathon Riddell}
\email{ppj.p.riddell@bham.ac.uk}
\affiliation{School of Physics and Astronomy, University of Birmingham, Edgbaston, Birmingham, B15 2TT, UK}
\affiliation{School of Physics and Astronomy, University of Nottingham, Nottingham, NG7 2RD, UK}
\affiliation{Centre for the Mathematics and Theoretical Physics of Quantum Non-Equilibrium Systems, University of Nottingham, Nottingham, NG7 2RD, UK}
\author{Nathan  Pagliaroli}
\email{npagliar@uwo.ca}
\affiliation{Department of Mathematics, Western University, N6A 3K7, Canada}
	\begin{abstract}
    In this article we investigate no-resonance conditions for quantum many body chaotic systems and random matrix models. No-resonance conditions are properties of the spectrum of a model, usually employed as a theoretical tool in the analysis of late time dynamics. The first order no-resonance condition holds when a spectrum is non-degenerate, while higher order no-resonance conditions imply sums of an equal number of energies are non-degenerate outside of permutations of the indices. This resonance condition is usually assumed to hold for quantum chaotic models. In this work we use several tests from random matrix theory to demonstrate that the  statistics of sums of eigenvalues, that are of interest to due to the no-resonance conditions, have Poisson statistics, and lack level repulsion. This result is produced for both a quantum chaotic Hamiltonian as well as the Gaussian Unitary Ensemble and Gaussian Orthogonal Ensemble. This implies some models may have violations of the no-resonance condition or "near" violations. We finish the paper by generalizing important bounds in quantum equilibration theory to cases where the no-resonance conditions are violated, and to the case of random matrix models. 
	\end{abstract}
	
		\maketitle

		One of the most ubiquitous observations in many body physics is the connection between the spectral statistics of many body quantum systems and that of random matrices. Quantum systems are not chaotic in the classical sense since unitary time evolution guarantees that the overlap between two states in time is constant. This excludes the classical notion of chaos in quantum systems for which we observe exponential sensitivity to small differences in initial conditions. However, their spectral statistics behave qualitatively differently if their corresponding classical limit is integrable or chaotic. If the classical limit is chaotic, the spectral statistics of the quantum Hamiltonian agree with the predictions of Random Matrix Theory (RMT) and we refer to these models as quantum chaotic \cite{berry87}. The notion of quantum chaos can be extended to quantum systems that do not have a well-defined classical limit \cite{d2016quantum}. 
  
  An extremely important property of the spectral statistics of a quantum chaotic   Hamiltonian is the presence of level-repulsion amongst neighbouring energies. This level-repulsion was first modeled for heavy atomic nuclei by Wigner using Gaussian ensembles of random matrices.  Since Wigner's work, it has been established that features of the spectrum of classically chaotic quantum
systems are accurately described by various ensembles of random matrices \cite{porter1965statistical,brody1981random,guhr1998random,Berry76,Berry77}.  The connection between the spectrum of quantum chaotic systems and random matrices has been well studied in single particle systems \cite{jalabert90,marcus92,Milner01,Friedman01,stockmann90,Sridhar91,
Moore94,Steck01,Hensinger01,Chaudhury09,Weinstein02,Zhang2022,Rigol202111,Vidmar202111}, along with many body systems \cite{Santos2010,Santos2010v2,Rigol10,Kollath_2010,Santos2012,Richter_2020,Atas2013,Atas_2013v2,Prosen2020}  and recently has seen a surge of interest in the case of circuit or periodically driven type models \cite{chan2018solution,Luca2014,Bertini2021,Bertini2018}. The first to extend Wigner's work were Dyson and Mehta in the series of papers \cite{dyson1962I,dyson1962II,dyson1962III,dysonmethaIV,dysonmehta1963V}. In particular, Dyson classified the three most immediately relevant ensembles: the Gaussian Unitary Ensemble, the Gaussian Orthogonal Ensemble, and the Gaussian Symplectic Ensemble in what is known as the ``threefold way" \cite{dyson1962threefold}. Of the most immediate interest to this work is the Gaussian Orthogonal Ensemble (GOE). The Bohigas, Giannoni, and Schmit (BGS) conjecture \cite{bohigas1984characterization} states that the GOE has the same level-spacing as a wide class of quantum systems with classical limits \cite{bohrdt2017scrambling,andreev1996quantum,muller2004semiclassical}. Let $E_{0} \leq  E_{1} \leq  E_{2}, ...$ be a sequence of unfolded energy eigenvalues of the GOE; then Wigner surmised the distribution of average consecutive level-spacings, that is the average of $s_{k} = E_{k+1} - E_{k}$ for all $k$ is 
  \begin{equation}\label{eq:WignerSurmise}
    p(s) =\frac{\pi s}{2}e^{-\pi^{2} s^{2}/4}.
\end{equation}
To see how to unfold a spectrum see Chapter 6  of \cite{mehta2004random} or \cite{Bruus1997}. It is important to note that Wigner's Surmise is an approximation \cite{mehta1960statistical}, of the actual distribution, originally derived in \cite{Jimbo1980DensityMO}. This was further simplified in terms of  Painlev\'e  transcendentals in\cite{forrester2000exact}.

In contrast to level-repulsion, if one considers the level-spacing of i.i.d. random variables, not only does one not see repulsion, but rather one sees Poisson statistics (sometimes referred to as neutral statistics \cite{imbrie2016many}), which has been used as a marker for non-chaotic systems \cite{d2016quantum}. For a proof of this fact we refer the reader to \cite{livan2018introduction}. In particular after unfolding the spacing of such systems, the distribution is Poisson
\begin{equation}\label{eq:Poissonspacing}
    p(s) = e^{-s}.
\end{equation}
The presence of level-repulsion and GOE spectral statistics is a hallmark test of quantum chaos, while Poisson statistics are associated with integrable or non-chaotic models.

A key consequence of the presence of level-repulsion is that the value of the probability density at zero is zero, meaning that we can assume our model has a non-degenerate spectrum. This observation is useful, for example, when considering dephasing arguments, which has recently been particularly popular in the quantum equilibration community \cite{alhambra2020time, Riddell2020,gogolin2016equilibration,masanes2013complexity,wilming2018equilibration,heveling2020compelling,Campos_Venuti_2010,Knipschild2020,carvalho2023equilibration}. If we consider the time-evolution of many dynamical functions under unitary dynamics, time-dependent terms in the series will often appear as the following:
\begin{equation}
     z \, e^{-i(E_{m}-E_{n})t},
\end{equation}
where $z$ is a complex number and $t$ is time. A typical example of such a series would be the expectation value of an observable for pure state time evolution. Terms such as these survive the infinite time average if and only if $E_{m} = E_{n}$. In the case of quantum chaotic Hamiltonians it is a safe assumption that any surviving term would imply that $m=n$, since we do not expect degeneracy due to the presence of level-repulsion. The degeneracies where $E_{m} = E_{n}$ and $m \not = n$ are examples of first order \textit{resonances}. However, in general,  dynamical functions can be more complex with terms such as
\begin{equation}\label{dynamical function}
     z \, e^{-i(E_{m_{1}}-E_{n_{1}}+ E_{m_{2}}-E_{n_{2}}+...)t}.
\end{equation}
Such terms appear in out-of-time-ordered correlators or other higher order correlation functions \cite{Riddell2020,riddell2021scaling,riddell2019out,Yoshida2019,Fortes2019,Shukla2022,Fortes_2020v2}.
  To discuss the terms that survive the infinite time average in equation \ref{dynamical function} we introduce the  $q$ no-resonance condition.
		
		\begin{definition} \label{def:genericspec} Let $H$ be a Hamiltonian with spectrum $H= \sum_j E_j \ket{E_j}\bra{E_j}$, and let $\Lambda_q,\Lambda'_q$ be two arbitrary sets of $q$ energy levels $\{ E_{j} \}$. $H$ satisfies the $q$ no-resonance condition if for all $\Lambda_q,\Lambda'_q$, the equality
\begin{equation}
    \sum_{j \in \Lambda_q } E_{j} =    \sum_{j \in \Lambda'_q } E_{j}
\end{equation}
implies that $\Lambda_q=\Lambda'_q$.
\end{definition}
By definition \ref{def:genericspec} the set of terms that satisfy the $q$ no-resonance condition are the minimum set of terms that survive the infinite time average as in equation \ref{dynamical function}. Terms that fall outside of definition \ref{def:genericspec} are referred to as $q$\textit{-resonances}. 
Typically in the literature it is suggested that quantum chaotic Hamiltonians satisfy definition \ref{def:genericspec} \cite{Mark2022,Srednicki99,Riddell2022}. 

This greatly  simplifies arguments involving infinite time averages in quantum chaotic models.  Despite this condition being somewhat common in the literature, studies only test this condition for the $q=1$ case where one finds level-repulsion governed by the Wigner-Dyson distribution \cite{Richter_2020, d2016quantum}. 

Of interest are the statistical properties of the sums and level-spacing of energy levels. We will refer to distribution of the sum of $q$ eigenvalues as the $q$-th order density of states and the corresponding level-spacing as the $q$-th order level-spacing . For the $q=2$ case, an explicit formula is known for the density of states \cite{khalkhali2022spectral}, but as far the authors can tell nothing is known about the second order level-spacing distribution. Note that a similar problem has been studied in the context of coupled kicked rotors, where the authors studied the transition from Poisson to RMT statistics \cite{Srivastava2016,Herrmann}. For zero coupling the problem is similar to testing the no-resonance condition for $q=2$, where the authors find Poisson statistics. In this article we go beyond the $q=2$ study and provide evidence that the Poisson statistics persist for $q\geq 2$. This result is non-trivial; for example, the problem for $q\geq 2$ is strikingly similar to studying free fermionic models at fixed particle number. Tight binding models for example have a large number of trivial degeneracies in the single particle spectrum resulting in a delta function around $s=0$ for the level spacing distribution.

		\section{Spectral statistics for a quantum chaotic Hamiltonian} \label{Sec:SpecStates}
		In this section we first investigate what the spectral statistics look like for a specific quantum chaotic model. In particular we study a Heisenberg type model with nearest- and next-nearest-neighbour interactions.
		\begin{align}
	{H} =  &\sum_{j=1}^L J_1 \left(  {S}_j^+ S_{j+1}^- + \text{h.c.}\right) + \gamma_1  \, {S}_j^Z    {S}_{j+1}^Z  \\ & +  J_2 \left(  {S}_j^+  {S}_{j+2}^- + \text{h.c.}\right)+ \gamma_2   {S}_j^Z  {S}_{j+2}^Z, \label{eq:hamiltonian}
\end{align}
where $(J_1,\gamma_1,J_2,\gamma_2) \nobreak = \nobreak (-1,1,-0.2,0.5)$ gives us a non-integrable model. This model has a free limit for $(J_1,0,0,0)$ and an interacting integrable limit for $(J_1,\gamma_1,0,0)$. Recently this model was confirmed to obey the eigenstate thermalization hypothesis \cite{LeBlond_2019}. We perform full spectrum exact diagonalization in the maximally symmetric sector of this model. In particular, this matrix conserves the total magnetization $m_z = \sum_j S_j^Z$, and is translation invariant. We choose to work in the sector such that $\langle m_z \rangle  = 0$ with quasi-momenta $k = 0$. This allows us to further diagonalize the model with the spatial reflection symmetry $P$ and the spin inversion symmetry $Z$. 

In this section we will focus on the spectral statistics for the cases $q=1$ as a benchmark, and $q = 2$, the first no-resonance condition that is unexplored in the literature. As we will show in the appendix, the behavior for $q > 2$ is qualitatively similar to $q=2$. First, let us establish that our model satisfies the usual tests for quantum chaos in the $q=1$ case. Perhaps the most common test is to investigate the level-spacing distribution $s_j = E_{j+1}-E_{j}$. The act of unfolding allows us to have a universal scale for the comparison of spectra of different Hamiltonians. The distribution of $s_j$ for a quantum chaotic model should exhibit Wigner-Dyson statistics. To unfold the spectrum we use Gaussian broadening. Namely we map our energies $E_k$ to $\epsilon_k$ as follows: 

\begin{equation}
    \epsilon_k = N(E_k),
\end{equation}

\begin{equation}
    N(E) = \int_{-\infty}^E \sum_k \frac{1}{\sigma_k \sqrt{2 \pi}} e^{-\frac{(e-E_k)^2}{2 \sigma_k^2}} de,
\end{equation}
where we use the same convention as in \cite{Bruus1997} and take 
\begin{equation}
    \sigma_k = 0.608 \alpha \Delta_k,
\end{equation}
where $\Delta = (E_{k+\alpha}-E_{k-\alpha})/{2\alpha}$ and we find that $\alpha = 20$ is quite suitable for our spectrum. 

\begin{figure}[h!]
\centering
\includegraphics[width=0.45\linewidth]{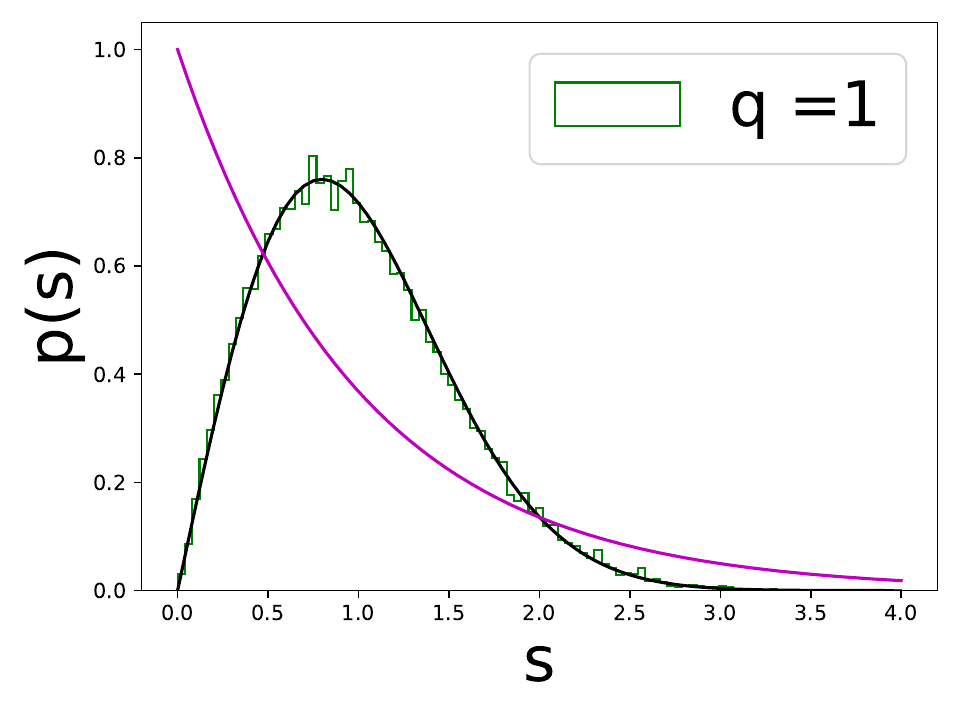}
\includegraphics[width=0.45\linewidth]{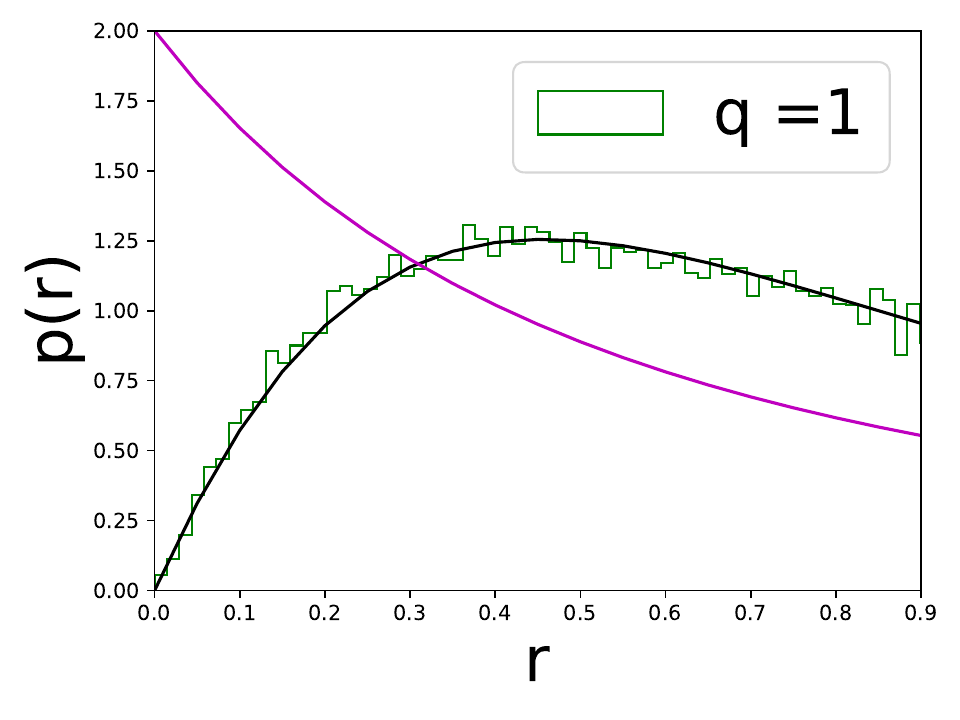}
\caption{ The left figure displays the first order level-spacing distribution for ($q=1$), $L = 24$ unfolded data, which looks approximately like a Wigner Surmise exhibiting level repulsion. In black we plot the Wigner Surmise and in purple we plot the Poisson distribution. In the right figure the results of the ratio test are shown for spectral statistics at $L = 24$.  }
\label{fig:q1}
\end{figure}
		Fig. \ref{fig:q1} demonstrates that our model for $q=1$ has level-repulsion and appears to have a level-spacing distribution that is  well approximated by the Wigner surmise. While this result shows us that our spectrum strongly resembles the predictions of RMT, the unfolding procedure is usually chosen to find such agreement, therefore it is desirable to perform a test that does not need unfolding. Such a test is given by investigating the distribution of ratios between successive gaps \cite{Rigol2021,oganesyan2007localization,atas2013distribution}.  We introduce the ratios
		\begin{equation} \label{eq:ratio}
		    r_j = \frac{\min \{s_j, s_{j+1} \}}{  \max \{s_j, s_{j+1} \}  },
		\end{equation}
		which tells us that $r_j \in [0,1]$. We emphasize that for the ratio test,   no unfolding procedure is required for $s_j$. This test can be done with the model's physical spectrum.
		For the GOE in \cite{atas2013distribution} it was analytically shown that the distribution of the $r_j$ for $3\times 3$ matrices is given by 
		\begin{equation} \label{eq:chaoticratio}
		    p(r) = \frac{27}{4} \frac{r+r^2}{\left( 1+r+r^2\right)^{\frac{5}{2}}}.
		\end{equation}
		If instead our energy levels were independent random variables, we would instead observe Poisson statistics, 
		
		\begin{equation} \label{eq:intratio}
		    p(r) = \frac{2}{(1+r)^2}.
		\end{equation}
		We see in Fig. \ref{fig:q1} (b) that our result experiences level-repulsion, agreeing with the distribution in equation \ref{eq:chaoticratio}.

		Next we consider the case for $q=2$. The spectrum we are now interested in is equivalent to the spectrum of the Hamiltonian,
		\begin{equation}
		    \hat{H}_2 = \hat{H} \otimes \mathbb{I} + \mathbb{I} \otimes \hat{H},
		\end{equation}
		which has the spectrum $\Lambda_{k,l} = E_k+E_l$. This construction introduces an unwanted symmetry in the spectrum of $\hat{H}_2$, namely that $\Lambda_{k,l} = \Lambda_{l,k}$, that is, the spectrum is invariant under permutations of the individual energies' indices. For $q=2$ this might be understood as a spatial reflection symmetry for a larger two component non-interacting system. Addressing this symmetry is simple. We only consider unique pairs of $(k,l)$, namely, we take $l>k$, where we also ignore the portion of the spectrum where $k=l$. Ignoring $k=l$ does not appear to significantly alter the results but allows us to eliminate trivial multiples of the $q=1$ spectrum. In fact, the contribution of the $q=l$ portion of the spectrum is 
vanishingly small compared to the total size of our spectrum. We further introduce a new index that orders the spectrum $\alpha = 1,2\dots$ such that $\Lambda_\alpha< \Lambda_{\alpha+1}$. With this new spectrum we can analyze the level-spacing and ratio distribution.

    \begin{figure}[h!]
    \centering
    \includegraphics[width=0.45\linewidth]{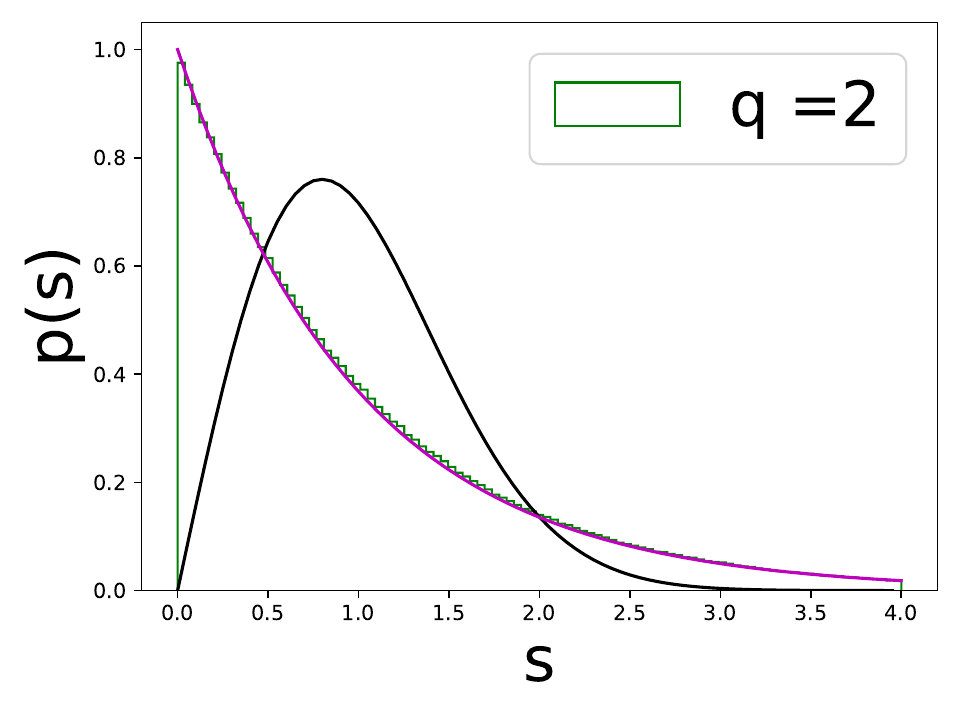}
\includegraphics[width=0.45\linewidth]{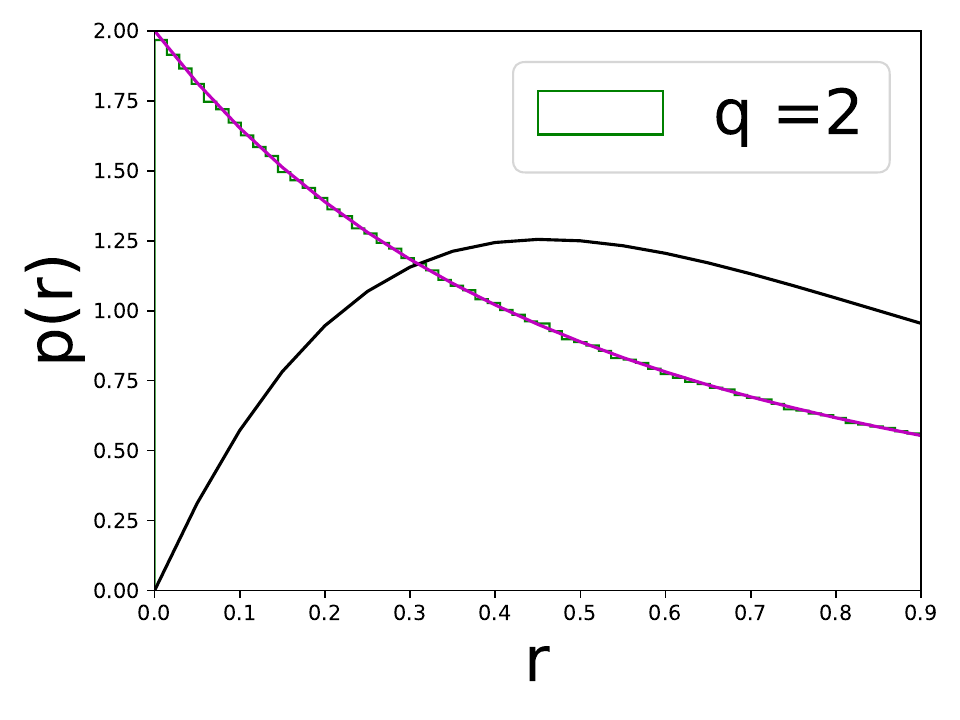}
\caption{ The left figure displays the second order level-spacing distribution ($q=2$) , $L = 20$ unfolded data, which looks approximately like a Poisson distribution. (b) The right figure displays the results of the ratio test for spectral statistics at $L = 20$ for $q=2$. In both plots we draw the GOE prediction in black and the independent random variable prediction (Poisson) in purple.   }
\label{fig:q2}
\end{figure}
 
Fig. \ref{fig:q2} indicates that the spectrum of $\hat{H}_2$ experiences Poissonstatistics. This is contrary to the $q=1$ case which has level-repulsion. Importantly this indicates that the spectrum of $\hat{H}_2$ behaves like an integrable model, and has gaps clustered around $s=0$. While this does not guarantee violations of the $q=2$ no-resonance theorem, it does make violations more likely. Likewise, we expect a large amount of pseudo violations such that $s_j = \Lambda_{j+1} - \Lambda_{j}\to 0$ to occur much faster than the average gap. This means that these violations would appear as resonances in the spectrum, unless very large time scales, potentially non-physically large, are considered. In light of this, results such as \cite{short2011equilibration,Riddell2022,Srednicki99,Mark2022} should be investigated to understand the effects of resonances. In the appendix \ref{sec:q34} we demonstrate that the Poisson statistics persist for higher values of $q$ and conjecture that Poisson statistics persists for all values of $q>1$. 

One further test we can perform is to test the actual average value of $r$ we observe in the ratio distribution. $\langle r \rangle = 2\ln 2 - 1 \approx 0.38629436112$ for Poisson systems and $\langle r \rangle = 4-2\sqrt{3} \approx 0.535898384$ for the GOE. Testing this quantity allows us to clearly observe convergence to the predictions of random matrix theory as a function of system size. We see this displayed in Fig. \ref{fig:ratioconverg}. In the right panel we see the test for $q=2$ which reveals a strong convergence in agreement with the Poisson predictions. The data at $L = 22$ gives $\langle r\rangle = 0.386294325894$, which confirms seven decimal points of convergence. Therefore, from the perspective of short range correlations in the spectrum we conclude that $\hat{H}_2$ obeys Poisson statistics.

    \begin{figure}[h!]
    \centering
    \includegraphics[width=0.45\linewidth]{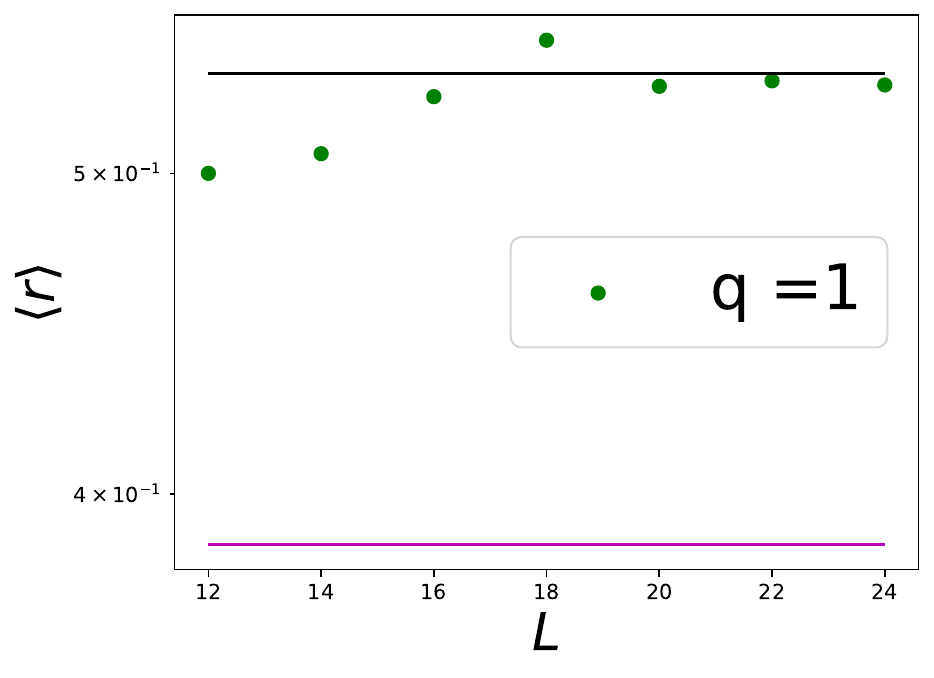}
\includegraphics[width=0.45\linewidth]{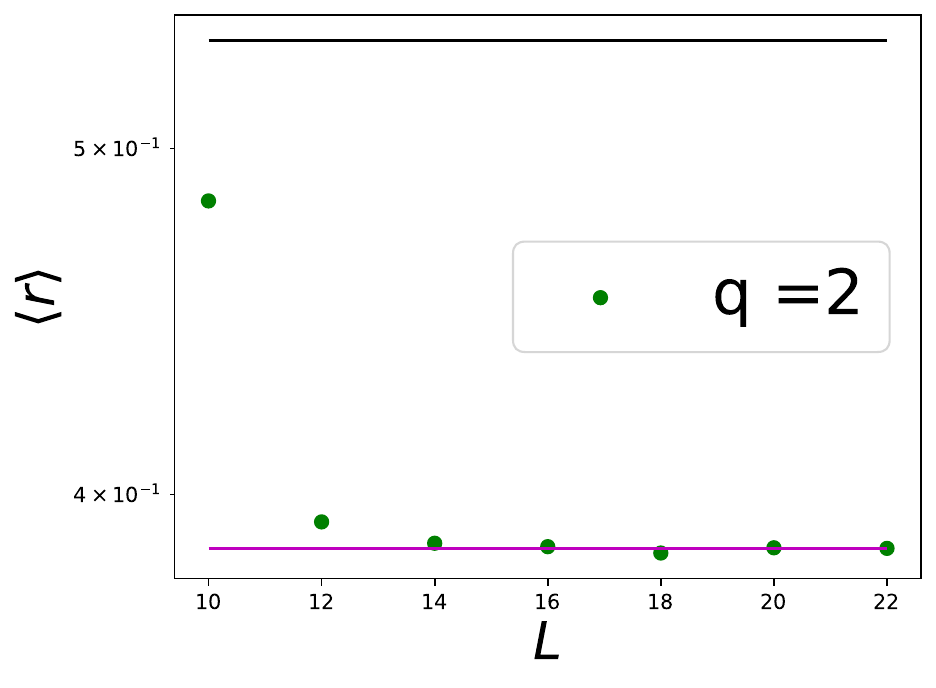}
\caption{ Plotted data for the convergence of $\langle r \rangle$ to RMT predictions.  In black we plot the GOE predictions and in purple we plot the corresponding Poisson prediction. Left: We see the $q=1$ data converge to the GOE prediction as a function of system size. Right: We see the $q=2$ data.   }
\label{fig:ratioconverg}
\end{figure}
 In appendix \ref{sec:RMTq23} we repeat our numerical studies but for random matrices, showing that our results from a quantum chaotic Hamiltonian agree with the results of RMT. Importantly our tests here are local tests on the spectrum. It is an open question if the symmetry resolved Hamiltonian $\hat{H}_2$ will still obey Poisson statistics for more complex tests such as investigating the spectral form factor \cite{bertini2021random,chan2018}.The non-symmetry resolved Hamiltonian corresponding to the $q$ no-resonance condition will not demonstrate Poisson statistics. This follows from the $q = 2$ case corresponding to the second moment of the spectral form factor for the GOE, which has a quadratic ramp \cite{mehta2004random}. We leave investigation of the spectral form factor to future work.  An important observation from this work is that when studying $q>1$ cases are never found where the spectral statistics are not Poisson. We conjecture that for $q>4$ the Poisson statistics persist, marking a stark contrast to the single particle $q$ no-resonance condition where, for typically studied models the condition is robustly violated \cite{riddell2024generic}.

We emphasize that the presence of Poisson statistics does not imply violations of the $q>1$ no-resonance condition. It does, however, imply the gaps in the spectrum of $\hat{H}_2$ cluster close to zero. If we investigate the probability of finding a gap within the range $0<s<\epsilon$, where $\epsilon$ is small, we have for the GOE, 
\begin{equation}
    \int_0^\epsilon \frac{\pi s}{2}e^{-\pi^{2} s^{2}/4} ds = \frac{1-e^{-\frac{-\pi^2 \epsilon^2}{4}}}{\pi} \approx \frac{\pi \epsilon^2}{4} - \frac{\pi^3 \epsilon^4}{32} \dots,
\end{equation}
so we see that the probability is proportional to $\epsilon^2$ for small gaps. On the contrary for the Poisson distribution one intuitively yields something much larger, 

\begin{equation}
     \int_0^\epsilon e^{-s} ds = \sinh \epsilon - \cosh \epsilon +1 \approx \epsilon - \frac{\epsilon^2}{2} \dots,
\end{equation}
giving us only linear scaling for small gaps. While both probabilities are of course small, the GOE is significantly smaller, giving one a significantly stronger case to assume definition \ref{def:genericspec} is satisfied in one's choice of chaotic model, including an absence of approximate resonances. In the case of Poisson statistics, one might expect to find one or many gaps that are essentially zero.

Infinite time averages are theoretical tools for which we average over times significantly longer than the Heisenberg time $\tau_H   \sim e^S$, where $S$ is the thermodynamic entropy at the appropriate energy $E$ \cite{Srednicki99}. The presence of essentially zero gaps will lead to terms $e^{i(E_k-E_{k-1})t}$ which are stationary on time scales proportional to $\tau_H$. Despite the presence of such violators, we expect the set of problematic gaps to be small relative to the total Hilbert space dimension. Since it is likely that some violations or cases that are indistinguishable from violations of definition \ref{def:genericspec} are inevitable, especially for cases using $q>1$, it is instructive to revisit past results keeping in mind a small number of violations (or approximate violations) will most likely be present. 
Below we discuss modifying key results in the field of quantum equilibration theory to accommodate the presence of violations of definition \ref{def:genericspec}.

		\section{Equilibration and recurrence} 
		\subsection{Physical models}
	    In this section we approach the problem of equilibration in light of our investigation of the higher order no-resonance conditions and the presence of Poisson statistics. First, let us review a basic setup. Consider a time independent system with the Hamiltonian $\hat{H}$ where we label the energy eigenbasis as $\hat{H}|E_k\rangle = E_k |E_k\rangle$. For simplicity, we take the spectrum of $\hat{H}$ to be discrete and finite. We will initialize our system in some pure state 
	    \begin{equation}
	        |\psi (t=0)\rangle = \sum_{m} c_m |E_k\rangle.
	    \end{equation}
	    To track equilibration, we study properties of the expectation value of an observable $\hat{A}$. This observable is general, but we demand that the largest single value $||A||$ is independent of system size or saturates to a finite value in the thermodynamic limit. In what follows we will assume our spectrum has level-repulsion, so that we may safely assume 
	    \begin{equation}
	        E_m = E_l \implies m = l. \end{equation}
	    
	    If our observable equilibrates, its finite time value $\langle \hat{A}(t) \rangle = \langle \psi(t)| \hat{A} | \psi(t) \rangle$ must relax to its infinite time average value, i.e. 
	    \begin{equation}
	        \Bar{A} = \lim_{T \to \infty} \frac{1}{T} \int_0^T\langle \hat{A}(t) \rangle dt =\lim_{T \to \infty} \frac{1}{T} \int_0^T \sum_{m,n}
	   \bar{c}_m c_n A_{m,n} e^{i(E_m-E_n)t} dt = \sum_m |c_m|^2 A_{m,m}. \end{equation}
	    $\bar{A}$ is usually written in terms of the diagonal ensemble $\omega = \sum_m |c_m|^2 |E_m\rangle \langle E_m|$ as $\bar{A} = \tr \left( \omega \hat{A}\right)$.
	  A typical quantity to study in quantum equilibration would be the variance of the expectation value around $\bar{A}$. This was studied and bounded in \cite{short2011equilibration,reimann2012equilibration},  assuming that the $q=2$ no-resonance condition was satisfied. The variance is written as 
	  
	  \begin{equation}
	      \mu_2 = \lim_{T \to \infty} \frac{1}{T} \int_0^T \left( \langle\hat{A}(t) \rangle - \bar{A} \right)^2 dt. 
	  \end{equation}
	  It was famously found in \cite{short2011equilibration} that this variance can be bounded by the purity of the diagonal ensemble 
	  \begin{equation} \label{eq:shortbound}
	      \mu_2 \leq ||A||^2 \tr \left( \omega^2 \right) .
	  \end{equation}
	  Note equation \ref{eq:shortbound} holds as a consequence of the $q=2$ no-resonance condition. The purity of the diagonal ensemble usually decays exponentially fast with respect to the system size (see for example Fig. 2 in \cite{Riddell2022}). If one assumes higher order $q$ no-resonance conditions, it was recently found that, for higher moments of equilibration, 
	  \begin{equation} \label{eq:kappa}
	      \mu_q = \lim_{T \to \infty} \frac{1}{T} \int_0^T \left( \langle\hat{A}(t) \rangle - \bar{A} \right)^qdt = \lim_{T\to \infty} \frac{1}{T} \int_0^T  \sum_{m_{1}\neq n_{1},...,m_{q}\neq n_{q}} \prod_{i=1}^{q}\left(A_{m_{i},n_{i}}\bar{c}_{m_{i}} c_{n_{i}}\right)e^{i(E_{m_{i}} -E_{n_{i}})t}dt,
	  \end{equation}
	  a similar bound can be found \cite{Riddell2022}, 
	  \begin{equation}
	      |\mu_q| \leq \left(q ||A|| \sqrt{\tr \left( \omega^2 \right)} \right)^q.
	  \end{equation}

	  In light of section \ref{Sec:SpecStates} and the presence of Poisson statistics for higher order $q$, these results should be updated to reflect the high probability of there being a violation of the $q$ no-resonance condition. 
	  \begin{theorem}
	  Suppose we have a model that has violations of the $q$ no-resonance condition. Then the moments $\mu_q$ can be bounded as
	  \begin{equation}
	      |\mu_q| \leq ||A||^q \left( q^q + \frac{\mathcal{N}_{q,L}}{2q}\right)\sqrt{\tr \left( \omega^2 \right)}^q,
	  \end{equation}
	  where $\mathcal{N}_{q,L}$ is the maximum number of times one $E_m$ appears in violations of the $q$ no-resonance condition for a given system size $L$. We call the $E_m$'s that appear in more than one violation of the resonance condition exceptional violators. 
	  \end{theorem}
	  \begin{proof}
	  Terms that contribute to $\mu_q$ are sums of energies that are equal. Let $\Lambda_q$ and $\Lambda_q'$ be sets of indices corresponding to particular energies 
	  
	  \begin{equation}
	      \sum_{m\in\Lambda_q} E_m = \sum_{m\in\Lambda_q'} E_m.
	  \end{equation}
	  The no-resonance condition picks out the trivial set of energies that satisfy this equality, which is when $\Lambda_q = \Lambda_q'$. These contributions were bounded in \cite{short2011equilibration,Riddell2022}. We collect the remaining violations in a set $\mathcal{S}$ and write, 
	  \begin{equation}
	      |\mu_q| \leq \left(q ||A|| \sqrt{\tr \left( \omega^2 \right)} \right)^q +\left| \sum_{\Lambda_q \in \mathcal{S}} \prod_{j=1}^q \bar{c}_{m_j} c_{n_j} A_{m_j,n_j} \right|,
	  \end{equation}
	  where we have identified $\Lambda_q \in \mathcal{S} = \{ m_j,n_j \}$. The second term can be bounded as follows. 
	  \begin{equation}
	      \left|\sum_{\Lambda_q \in \mathcal{S}} \prod_{j=1}^q \bar{c}_{m_j} c_{n_j} A_{m_j,n_j} \right| \leq ||A||^q \sum_{\Lambda_q \in \mathcal{S}} \prod_{j=1}^q |c_{m_j}|| c_{n_j}|.
	  \end{equation}
	  Since all $|c_{m_j}|$ are positive, we may use the inequality of arithmetic and geometric means, giving
	  \begin{equation}
	      \leq \frac{||A||^q}{2q} \sum_{\Lambda_q \in \mathcal{S}} \sum_{j=1}^q \left( |c_{m_j}|^{2q} + |c_{n_j}|^{2q} \right).
	  \end{equation}
	  We know that $\tr\left( \omega^q \right) = \sum_m |c_m|^{2q}$.  Define $\mathcal{N}_{q,L} = \max_{E_{m} } | \{ X \in \mathcal{S} | E_{m} \in X\}|$. Then it follows that an individual $ |c_{m_j}|^{2q}$ contributes at most $\mathcal{N}_{q,L}$ times, we have that
	  \begin{equation}
	      \leq \frac{||A||^q \mathcal{N}_{q,L}}{2q} \tr \left( \omega^q \right),
	  \end{equation}
   where $N_{q,L} = \max_{}$
	  We lastly recall that $\tr\left( \omega^q \right) \leq \tr\left( \omega^2 \right)^{q/2} $, which completes the proof. 
	  \end{proof}

	  Accommodating the presence of degenerate gaps for the $q=2$ case has been considered before in \cite{ShortFarrelly11}.  Our bound reads 
	  
	  \begin{equation}
	      |\mu_2| \leq ||A||^2 \left( 4+\frac{\mathcal{N}_{2,L}}{4}\right) \tr\left( \omega^2 \right).
	  \end{equation}
	 Alternatively, one can likewise write \cite{ShortFarrelly11} as 
	  
	  \begin{equation}
	      |\mu_2| \leq N(\epsilon) ||A||^2\tr\left( \omega^2 \right),
	  \end{equation}
	  where $N(\epsilon)$ is the maximum number of energy gaps in any interval for $\epsilon > 0$, i.e.
	  
	 \begin{equation}
	     N(\epsilon) = \max_E | \{ (k,l) | \enspace  E_k - E_l \in [E,E+\epsilon) \}|.
	 \end{equation}
	  One can recover the maximum degeneracy of the gaps by considering $\lim_{\epsilon \to 0^+} N(\epsilon)$. In the limit of non-degenerate gaps these bounds are identical, and only differ by a constant factor for a small number of degeneracies in the gaps. Our result might in theory give better constant factors than the result in \cite{ShortFarrelly11}, however $N(\epsilon)$ is likely a more intuitive quantity and easier to work with numerically.

	  We next wish to understand the properties of $\mathcal{N}_{q,L}$, which in practice is challenging to study numerically. The worst scaling it could have is the total number of violations, i.e. $0\leq\mathcal{N}_{q,L} \leq |\mathcal{S}|$. As we have noted earlier, the presence of Poisson statistics does not imply $|\mathcal{S}| >0$. 
	  An easy property to understand is that if $\mathcal{N}_{q,L}\geq 2$ this implies at the very least that $\mathcal{N}_{q+1,L} \geq 1$. To see this consider $q = 2$ for an exceptional violator $E_m$ that appears at least twice. We might have $E_m$ as an exceptional violator as 
	  \begin{equation}
	      E_m + E_n = E_p +E_l, \enspace E_m + E_k = E_r + E_h.
	  \end{equation}
	    For $q = 3$ a violation of the no-resonance condition is 
	    \begin{equation}
	        E_p + E_l + E_k = E_r + E_h +E_n. 
	    \end{equation}
		Despite two exceptional violations for $q = 2$ implying at least one for $q = 3$, this does not imply $\mathcal{N}_{q,L}$ is decreasing in $q$. To estimate the size of $\mathcal{N}_{q,L}$ we can attempt to quantify the expected or average behavior of the quantity. First, let us assume we randomly generated the set $S$. We will assume that the indices which appear are uniformly generated, so each element of $S$ can be understood to be a tuple of $2q$ indices, $(m_1, \dots m_{2q})$. Despite this construction the indices aren't statistically independent due to restrictions on acceptable tuples in $S$. For example the indices in the tuple cannot be equal to each other under our assumptions. Despite this, in the large $L$ limit this dependence cannot affect results due to the smallness of $q$ and the corresponding exponential nature of the number of possible indices $2^L$. We can therefore focus on the average behavior of the first index of each tuple $m_1$. Our goal will be to predict the average number of times the first index $m_1$ ends up being the same index for different tuples in $S$. The same index can appear in different tuples at most $|S|$ times, and thus we wish to compute 
\begin{equation}
     \langle \mathcal{N}_{q,L} \rangle= \sum_{n=1}^{|S|} n p(n),
\end{equation}
where $p(n)$ is the probability of the same index appearing $n$ times. The total number of configurations possible for the first index of each tuple in $S$ is $2^{|S|L}$ and therefore we must simply count the number of configurations where $n$ copies of the same $m_1$ appear. This is given by 
\begin{equation}
    {|S|\choose n} 2^{L \left( |S|-n \right)},
\end{equation}
which gives the following formula for our expected value 
\begin{equation}
     \langle \mathcal{N}_{q,L} \rangle= \sum_{n=0}^{|S|} \frac{n {|S|\choose n}}{2^{Ln}} = \frac{|S|}{2^{L}} (2^{-L}+1)^{|S|-1}.
\end{equation}

We now have some special limiting cases to consider. Suppose that $|S| \propto c \,2^{L}$ for some constant $c$. Then the expected value of $\langle \mathcal{N}_{q,L} \rangle$ is $c\,e^{c}$ as $L$ goes to infinity. However, if $|S|$ has sub exponential growth, for example if it scales as $L$, then the expected value goes to zero for large system size as $\mathcal{O}(L/2^{L})=  \mathcal{O}(|S|/2^{L})$. Therefore we expect that for most cases, even with modest violations of the no-resonance condition we expect $\lim_{L \to \infty} N_{q,L}$ to be finite and quite small. This conclusion is further amplified by the findings in the previous section where we found the level spacing statistics of sums of energies to result in Poisson statistics. These observations combined imply that for the vast majority of cases the corrections due to resonances should be small in the case studied here. It would be interesting to investigate other contexts as this might not be the case.



		\subsection{A Random Matrix Theory Approach}
  
    In this section we will show how one could compute $\mu_{q}$ for the GUE and GOE with an unfolded spectrum in the large $N$ limit. We can rewrite equation \ref{eq:kappa} for finite $T$ as 
   \begin{equation}
     \mu_{q}(T,N)  = \sum_{i_{1}\not = j_{1},...,i_{q}\not = j_{q} } \left( \prod_{k=1}^{q}\overline{c}_{i_{k}}c_{j_{k}} \langle i_{k}|A|j_{k}\rangle \right)   e^{i\sum_{k=1}^{q}(\lambda_{i_{k}}-\lambda_{j_{k}})t }  , 
   \end{equation}
   where the eigenvalues are unfolded and from some matrix ensemble. Define the $n$-level  form factor of a matrix ensemble as 
\begin{equation}
      \mathcal{R}_{n} (t)= \mathbb{E}\left[  \frac{1}{N^{2n}} e^{i t\sum_{k=1}^{n}\lambda_{i_{k}}-\lambda_{j_{k}}}\right],
\end{equation}
and the connected component of the $n$-level form factor $\mathcal{K}_{n}$ with the moment-cumulant relations
\begin{equation}\label{eq:moment-cumulant}
    \mathcal{R}_{n}(t) =\sum_{\pi \vdash [n]} (-1)^{n-|\pi|} \prod_{B \in \pi} \mathcal{K}_{|B|}(t)
\end{equation}
where the sum is over partitions of the set $[n]:=\{1,2,...,n\}$, $|\pi|$ denotes the number of blocks in the partition, the product is over blocks of the partition, and $|B|$ denotes the number of element of a block $B$. See chapter 6 of \cite{mehta1960statistical} for details.
 In the following theorem we apply classical results on the $n$-level form factors from \cite{mehta2004random} to study Hamiltonians with no time-reversal symmetry \cite{mehta1960statistical,bohigas1995chaotic}. Note that these are not exactly spectral form factors, but rather the expectation value of their summand. Also note that the $t$ in above (connected) $n$-level  form factors is in units of Heisenberg time $t_{H} = 2 \pi \overline{\rho}$, where $\overline{\rho}$ is the average of the spectral density function. Thus, we will need to re-scale to physical time before taking our time average.  
\begin{theorem}
  Consider a model with a GUE unfolded spectrum. Then 
the infinite time average of 
of the expected value of $\mu_{2}(T,N)$  goes to zero as $1/T^2$ as $T$ and $N$ go to infinity at the same rate.
\end{theorem}
\begin{proof}
Consider 
 \begin{align}
    \lim_{T \rightarrow \infty}\frac{1}{T}\int_{0}^{T}\lim_{N \rightarrow \infty}\mathbb{E}\left[\frac{1}{N^4} \mu_{2}(T,N)\right] dt&=\lim_{T \rightarrow \infty}\frac{1}{T}\int_{0}^{T}\lim_{N \rightarrow \infty}\mathbb{E}\left[\frac{1}{N^{4}}\sum_{i_{1}\not =j_{1},i_{2}\not = j_{2}} \left( \prod_{k=1}^{q}\overline{c}_{i_{k}}c_{j_{k}} \langle i_{k}|A|j_{k}\rangle \right) e^{i\sum_{k=1}^{2}(\lambda_{i_{k}}-\lambda_{j_{k}})t }  \right ] dt\\
     & =\sum_{i_{1}\not =j_{1},i_{2}\not = j_{2}} \left( \prod_{k=1}^{2}\overline{c}_{i_{k}}c_{j_{k}} \langle i_{k}|A|j_{k}\rangle \right) \frac{1}{T}\int_{0}^{T}\lim_{N \rightarrow \infty}\frac{1}{N^{4}}\mathbb{E}\left[e^{i\sum_{k=1}^{2}(\lambda_{i_{k}}-\lambda_{j_{k}})t }  \right ] dt\\
\end{align}
Note that we can pull the expectation value inside because the terms in our product dependent only on the eigenvectors of the GUE, which are independent of its eigenvalues. This follows from the fact that the joint probability distribution $P(H)$ for the entries of the GUE is invariant under the action of the unitary group. For more details see the proof of Weyl's integration formula \cite{deift2000orthogonal}. Note that the $1/N^4$ factor is the proper normalization for the GUE expectation value of such an observable. We can split this sum into four sums, up to the symmetry of the summand, where either 1) $i_{1} =i_{2}$ and $j_{1} =j_{2}$, 2) $i_{1} \not=i_{2}$ and $j_{1} =j_{2}$, 3) $i_{1} =i_{2}$ and $j_{1} \not=j_{2}$, or 4) $i_{1} \not=i_{2}$ and $j_{1} \not=j_{2}$. In case 1) each term is of the form
\begin{align}
    \lim_{N \rightarrow \infty}\mathbb{E}\left[e^{2i(\lambda_{1}-\lambda_{2})t }  \right ] &=  \lim_{N \rightarrow \infty}\mathbb{E}\left[e^{2i\pi (\lambda_{1}-\lambda_{2})\left(\frac{2 t}{\pi}\right) }  \right ] =\mathcal{K}_{1}(2 t/\pi)^2 -\mathcal{K}_{2}(2 t/\pi)\approx \frac{1}{N} - \frac{1}{N}\begin{cases} 1
         - \left |\frac{2 t}{\pi} \right| & |t| \leq \frac{\pi}{2}\\
        0 & |t| \geq \frac{\pi}{2}.
    \end{cases}
\end{align}
which is just a rescaling of the usual spectral form factor for the GUE \cite{cipolloni2023spectral}. 

In case 2) each term is of the form 
    \begin{equation}
    \lim_{N \rightarrow \infty}\mathbb{E}\left[e^{i(\lambda_{1}+\lambda_{2}-2\lambda_{3})t }  \right ].
\end{equation}  In case 3) each term is of the form
    \begin{equation}
    \lim_{N \rightarrow \infty}\mathbb{E}\left[e^{i(2\lambda_{1}-\lambda_{2}-\lambda_{3})t }  \right ].
\end{equation}
Lastly, in  case 4) each term is of the form
    \begin{equation}
    \lim_{N \rightarrow \infty}\mathbb{E}\left[e^{i(\lambda_{1}+\lambda_{2}-\lambda_{3}-\lambda_{4})t }  \right ].
\end{equation}
Note that evaluating all cases is possible but unnecessary since only the contribution from case 4) contribute in the large $N$ limit. This follows from noting that there will be $N(N-1)$ identical many integrals for case 1), $N(N-1)(N-2)$ many for cases 2) and 3), and $N(N-1)(N-2)(N-3)$ for case 4). Hence, only terms in case 4) will contribute in the limit.

In what follows we will use formulae 6.2.18 and 6.2.19 from \cite{mehta2004random} to evaluate all $\mathcal{K}$.
First we must apply the moment-cumulant formula \ref{eq:moment-cumulant} to write
\begin{align}
    \lim_{N \rightarrow \infty}\mathbb{E}\left[e^{i(x+y-z-w)t }  \right ] &= -\mathcal{K}_{4}( t/\pi) + 4 \mathcal{K}_{3}( t/\pi) \mathcal{K}_{1}( t/\pi) - 3 \mathcal{K}_{2}( t/\pi)\mathcal{K}_{1}( t/\pi)^2 + \mathcal{K}_{1}^{4} (t/\pi)\\
    &\approx-\mathcal{K}_{4}(t/\pi) + \frac{4}{\sqrt{N}} \mathcal{K}_{3}( t/\pi)  - \frac{3}{N} \mathcal{K}_{2}( t/\pi) + \frac{1}{N^2}.
\end{align}

Now define the function 
  \begin{equation}
       f(x) =\begin{cases}
           1 & |x| \leq \frac{1}{2} \\
           0 & \text{otherwise}.
       \end{cases}
   \end{equation}
   then we can evaluate
\begin{align}
    K_{3}( 2 t/\pi ) &\approx \frac{2}{N^{3/2}} \int_{\infty}^{\infty}  f( 2\pi x)f( 2\pi (x+t)) f ( 2 \pi( x + 2 t) )dx\\
    &=  \frac{2}{N^{3/2}} \begin{cases}
           1- \frac{3}{\pi}|t| & 0 \leq t \leq \frac{\pi}{ 3} \\
           0 & |t|\geq   \frac{\pi}{ 3},
       \end{cases}
  \
\end{align}
and
\begin{align}
    K_{4}( 2 t/\pi ) &\approx \frac{6}{N^{2}} \int_{\infty}^{\infty} f( 2\pi x)f( 2\pi (x+t)) f ( 2 \pi( x + 2 t)) f(2 \pi (x + 3t))dx\\
    &= \frac{6}{N^{2}} \begin{cases}
           1-\frac{4}{\pi}| t | & 0 \leq t \leq \frac{\pi}{ 4} \\
           0 & |t|\geq   \frac{\pi }{ 4}.
       \end{cases}
\end{align}

Now as mentioned before, we need to convert our $t$ from units of Hesienberg time to physical time. For the GUE our Heisenberg time is $t_{H} = 2 \pi \overline{\rho} = 2 \pi (N/4) = \frac{\pi}{2N}$. Putting all this together we have that in physical time units

\begin{align}
   & \lim_{N \rightarrow \infty}\mathbb{E}\left[e^{i(x+y-z-w)t }  \right ] 
    \approx\\
    &\frac{1}{N^2}\left(-6 \begin{cases}
           1-\frac{8}{\pi^2 N}| t | & 0 \leq t \leq \frac{N\pi^2}{ 8} \\
           0 & |t|\geq   \frac{N\pi^2}{ 8} ,
       \end{cases}+ 8\begin{cases}
           1- \frac{6}{\pi^2 N}|t| & 0 \leq t \leq \frac{N\pi^2}{ 6} \\
           0 & |t|\geq   \frac{N\pi^2}{ 6},
       \end{cases}- 3\begin{cases}
        1 -    \frac{4}{\pi^2 N} \left |t\right| & |t| \leq \frac{\pi^2 N}{4}\\
        0 & |t| \geq \frac{\pi^2 N}{4}.
    \end{cases} + 1\right)\\
\end{align}

Now for any large $T$ such that $T> \frac{N \pi^2}{2}$, we have that the above quantity has a time average of 
\begin{equation}
  \approx \frac{1}{T}\left(-\frac{6\pi^2}{16} \frac{1}{N} + \frac{3 \pi^2}{4} \frac{1}{N} - \frac{3\pi^2}{8}\frac{1}{N}+\frac{T}{N^2}\right).
\end{equation}
The result then follows.
\end{proof}

 The idea of this proof could be applied to higher moments as well as other invariant ensembles. Usually the GUE and GOE are of interest, but progress has been made studying the spectral form factor for other matrix ensembles. For example see \cite{forrester2021differential,forrester2021quantifying,cipolloni2023spectral}. The $q$-level spectral form factor can be computed explicitly, but it is a computationally heavy task. For reference see \cite{liu2018spectral}, where it is computed but for ensembles that are not unfolded.

  As we demonstrate in appendix \ref{sec:RMTq23}, the spectrum of the random matrix Hamiltonian likewise experiences Poisson statistics for $q\geq 2$. However, despite the presence of Poisson statistics, the above RMT result indicates that we still should expect $\mu_q \to 0$ indicating equilibration of our observable.

\section{Conclusion}

In this work we have explored spectral statistics of chaotic Hamiltonians, namely the statistics surrounding sums of energies. We found that despite being chaotic, sums of energies displayed Poisson statistics instead of Wigner-Dyson statistics. This was demonstrated numerically for both a chaotic spin Hamiltonian and the GOE. The presence of Poisson statistics leads one to believe that accounting for potential degeneracies or ``resonances" in infinite time averages of some dynamical quantities is an interesting follow-up challenge to previous work. We applied this observation to the theory of equilibration where we generalized known bounds to accommodate for degeneracies. Assuming the number of degeneracies is not exponentially large in system size, we demonstrated that the the bounds can be easily generalized to accommodate the presence of resonances. We further used techniques from RMT to prove that for the GUE moments of equilibration go to zero in the thermodynamic limit.

\section{Acknowledgements}
J.R. would like to thank Bruno Bertini, Marcos Rigol and Alvaro Alhambra for fruitful conversations. J.R. would like to extend special thanks in particular to Bruno who gave valuable feedback at various stages of the project. J.R. acknowledges the support of Royal Society through the University Research Fellowship No. 201101. N.P. acknowledges support from the Natural Sciences and Engineering Research Council of Canada (NSERC). 

\subsection*{Data availability} The code to reproduce the the figures of this article will be shared upon request.

\subsection*{Conflict of interest}The authors have no competing interests to declare that pertain to the content of this article.

		\bibliography{references2}

\begin{thebibliography}{88}%
\makeatletter
\providecommand \@ifxundefined [1]{%
 \@ifx{#1\undefined}
}%
\providecommand \@ifnum [1]{%
 \ifnum #1\expandafter \@firstoftwo
 \else \expandafter \@secondoftwo
 \fi
}%
\providecommand \@ifx [1]{%
 \ifx #1\expandafter \@firstoftwo
 \else \expandafter \@secondoftwo
 \fi
}%
\providecommand \natexlab [1]{#1}%
\providecommand \enquote  [1]{``#1''}%
\providecommand \bibnamefont  [1]{#1}%
\providecommand \bibfnamefont [1]{#1}%
\providecommand \citenamefont [1]{#1}%
\providecommand \href@noop [0]{\@secondoftwo}%
\providecommand \href [0]{\begingroup \@sanitize@url \@href}%
\providecommand \@href[1]{\@@startlink{#1}\@@href}%
\providecommand \@@href[1]{\endgroup#1\@@endlink}%
\providecommand \@sanitize@url [0]{\catcode `\\12\catcode `\$12\catcode `\&12\catcode `\#12\catcode `\^12\catcode `\_12\catcode `\%12\relax}%
\providecommand \@@startlink[1]{}%
\providecommand \@@endlink[0]{}%
\providecommand \url  [0]{\begingroup\@sanitize@url \@url }%
\providecommand \@url [1]{\endgroup\@href {#1}{\urlprefix }}%
\providecommand \urlprefix  [0]{URL }%
\providecommand \Eprint [0]{\href }%
\providecommand \doibase [0]{http://dx.doi.org/}%
\providecommand \selectlanguage [0]{\@gobble}%
\providecommand \bibinfo  [0]{\@secondoftwo}%
\providecommand \bibfield  [0]{\@secondoftwo}%
\providecommand \translation [1]{[#1]}%
\providecommand \BibitemOpen [0]{}%
\providecommand \bibitemStop [0]{}%
\providecommand \bibitemNoStop [0]{.\EOS\space}%
\providecommand \EOS [0]{\spacefactor3000\relax}%
\providecommand \BibitemShut  [1]{\csname bibitem#1\endcsname}%
\let\auto@bib@innerbib\@empty
\bibitem [{\citenamefont {Berry}(1987)}]{berry87}%
  \BibitemOpen
  \bibfield  {author} {\bibinfo {author} {\bibfnamefont {M.~V.}\ \bibnamefont {Berry}},\ }\href {https://doi.org/10.1098/rspa.1987.0109} {\bibfield  {journal} {\bibinfo  {journal} {Proc. R. Soc. A}\ }\textbf {\bibinfo {volume} {413}},\ \bibinfo {pages} {183} (\bibinfo {year} {1987})}\BibitemShut {NoStop}%
\bibitem [{\citenamefont {D'Alessio}\ \emph {et~al.}(2016)\citenamefont {D'Alessio}, \citenamefont {Kafri}, \citenamefont {Polkovnikov},\ and\ \citenamefont {Rigol}}]{d2016quantum}%
  \BibitemOpen
  \bibfield  {author} {\bibinfo {author} {\bibfnamefont {L.}~\bibnamefont {D'Alessio}}, \bibinfo {author} {\bibfnamefont {Y.}~\bibnamefont {Kafri}}, \bibinfo {author} {\bibfnamefont {A.}~\bibnamefont {Polkovnikov}}, \ and\ \bibinfo {author} {\bibfnamefont {M.}~\bibnamefont {Rigol}},\ }\href {\doibase 10.1080/00018732.2016.1198134} {\bibfield  {journal} {\bibinfo  {journal} {Advances in Physics}\ }\textbf {\bibinfo {volume} {65}},\ \bibinfo {pages} {239} (\bibinfo {year} {2016})}\BibitemShut {NoStop}%
\bibitem [{\citenamefont {Porter}(1965)}]{porter1965statistical}%
  \BibitemOpen
  \bibfield  {author} {\bibinfo {author} {\bibfnamefont {C.~E.}\ \bibnamefont {Porter}},\ }\href@noop {} {\emph {\bibinfo {title} {Statistical theories of spectra: fluctuations}}},\ \bibinfo {type} {Tech. Rep.}\ (\bibinfo {year} {1965})\BibitemShut {NoStop}%
\bibitem [{\citenamefont {Brody}\ \emph {et~al.}(1981)\citenamefont {Brody}, \citenamefont {Flores}, \citenamefont {French}, \citenamefont {Mello}, \citenamefont {Pandey},\ and\ \citenamefont {Wong}}]{brody1981random}%
  \BibitemOpen
  \bibfield  {author} {\bibinfo {author} {\bibfnamefont {T.~A.}\ \bibnamefont {Brody}}, \bibinfo {author} {\bibfnamefont {J.}~\bibnamefont {Flores}}, \bibinfo {author} {\bibfnamefont {J.~B.}\ \bibnamefont {French}}, \bibinfo {author} {\bibfnamefont {P.}~\bibnamefont {Mello}}, \bibinfo {author} {\bibfnamefont {A.}~\bibnamefont {Pandey}}, \ and\ \bibinfo {author} {\bibfnamefont {S.~S.}\ \bibnamefont {Wong}},\ }\href@noop {} {\bibfield  {journal} {\bibinfo  {journal} {Reviews of Modern Physics}\ }\textbf {\bibinfo {volume} {53}},\ \bibinfo {pages} {385} (\bibinfo {year} {1981})}\BibitemShut {NoStop}%
\bibitem [{\citenamefont {Guhr}\ \emph {et~al.}(1998)\citenamefont {Guhr}, \citenamefont {M{\"u}ller-Groeling},\ and\ \citenamefont {Weidenm{\"u}ller}}]{guhr1998random}%
  \BibitemOpen
  \bibfield  {author} {\bibinfo {author} {\bibfnamefont {T.}~\bibnamefont {Guhr}}, \bibinfo {author} {\bibfnamefont {A.}~\bibnamefont {M{\"u}ller-Groeling}}, \ and\ \bibinfo {author} {\bibfnamefont {H.~A.}\ \bibnamefont {Weidenm{\"u}ller}},\ }\href@noop {} {\bibfield  {journal} {\bibinfo  {journal} {Physics Reports}\ }\textbf {\bibinfo {volume} {299}},\ \bibinfo {pages} {189} (\bibinfo {year} {1998})}\BibitemShut {NoStop}%
\bibitem [{\citenamefont {Berry}\ and\ \citenamefont {Tabor}(1976)}]{Berry76}%
  \BibitemOpen
  \bibfield  {author} {\bibinfo {author} {\bibfnamefont {M.~V.}\ \bibnamefont {Berry}}\ and\ \bibinfo {author} {\bibfnamefont {M.}~\bibnamefont {Tabor}},\ }\href {https://doi.org/10.1098/rspa.1976.0062} {\bibfield  {journal} {\bibinfo  {journal} {Proc. R. Soc. A}\ }\textbf {\bibinfo {volume} {349}},\ \bibinfo {pages} {101} (\bibinfo {year} {1976})}\BibitemShut {NoStop}%
\bibitem [{\citenamefont {Berry}\ and\ \citenamefont {Tabor}(1977)}]{Berry77}%
  \BibitemOpen
  \bibfield  {author} {\bibinfo {author} {\bibfnamefont {M.~V.}\ \bibnamefont {Berry}}\ and\ \bibinfo {author} {\bibfnamefont {M.}~\bibnamefont {Tabor}},\ }\href {https://doi.org/10.1098/rspa.1977.0140} {\bibfield  {journal} {\bibinfo  {journal} {Proc. R. Soc. A}\ }\textbf {\bibinfo {volume} {356}},\ \bibinfo {pages} {375} (\bibinfo {year} {1977})}\BibitemShut {NoStop}%
\bibitem [{\citenamefont {Jalabert}\ \emph {et~al.}(1990)\citenamefont {Jalabert}, \citenamefont {Baranger},\ and\ \citenamefont {Stone}}]{jalabert90}%
  \BibitemOpen
  \bibfield  {author} {\bibinfo {author} {\bibfnamefont {R.~A.}\ \bibnamefont {Jalabert}}, \bibinfo {author} {\bibfnamefont {H.~U.}\ \bibnamefont {Baranger}}, \ and\ \bibinfo {author} {\bibfnamefont {A.~D.}\ \bibnamefont {Stone}},\ }\href {\doibase 10.1103/PhysRevLett.65.2442} {\bibfield  {journal} {\bibinfo  {journal} {Phys. Rev. Lett.}\ }\textbf {\bibinfo {volume} {65}},\ \bibinfo {pages} {2442} (\bibinfo {year} {1990})}\BibitemShut {NoStop}%
\bibitem [{\citenamefont {Marcus}\ \emph {et~al.}(1992)\citenamefont {Marcus}, \citenamefont {Rimberg}, \citenamefont {Westervelt}, \citenamefont {Hopkins},\ and\ \citenamefont {Gossard}}]{marcus92}%
  \BibitemOpen
  \bibfield  {author} {\bibinfo {author} {\bibfnamefont {C.~M.}\ \bibnamefont {Marcus}}, \bibinfo {author} {\bibfnamefont {A.~J.}\ \bibnamefont {Rimberg}}, \bibinfo {author} {\bibfnamefont {R.~M.}\ \bibnamefont {Westervelt}}, \bibinfo {author} {\bibfnamefont {P.~F.}\ \bibnamefont {Hopkins}}, \ and\ \bibinfo {author} {\bibfnamefont {A.~C.}\ \bibnamefont {Gossard}},\ }\href {\doibase 10.1103/PhysRevLett.69.506} {\bibfield  {journal} {\bibinfo  {journal} {Phys. Rev. Lett.}\ }\textbf {\bibinfo {volume} {69}},\ \bibinfo {pages} {506} (\bibinfo {year} {1992})}\BibitemShut {NoStop}%
\bibitem [{\citenamefont {Milner}\ \emph {et~al.}(2001)\citenamefont {Milner}, \citenamefont {Hanssen}, \citenamefont {Campbell},\ and\ \citenamefont {Raizen}}]{Milner01}%
  \BibitemOpen
  \bibfield  {author} {\bibinfo {author} {\bibfnamefont {V.}~\bibnamefont {Milner}}, \bibinfo {author} {\bibfnamefont {J.~L.}\ \bibnamefont {Hanssen}}, \bibinfo {author} {\bibfnamefont {W.~C.}\ \bibnamefont {Campbell}}, \ and\ \bibinfo {author} {\bibfnamefont {M.~G.}\ \bibnamefont {Raizen}},\ }\href {\doibase 10.1103/PhysRevLett.86.1514} {\bibfield  {journal} {\bibinfo  {journal} {Phys. Rev. Lett.}\ }\textbf {\bibinfo {volume} {86}},\ \bibinfo {pages} {1514} (\bibinfo {year} {2001})}\BibitemShut {NoStop}%
\bibitem [{\citenamefont {Friedman}\ \emph {et~al.}(2001)\citenamefont {Friedman}, \citenamefont {Kaplan}, \citenamefont {Carasso},\ and\ \citenamefont {Davidson}}]{Friedman01}%
  \BibitemOpen
  \bibfield  {author} {\bibinfo {author} {\bibfnamefont {N.}~\bibnamefont {Friedman}}, \bibinfo {author} {\bibfnamefont {A.}~\bibnamefont {Kaplan}}, \bibinfo {author} {\bibfnamefont {D.}~\bibnamefont {Carasso}}, \ and\ \bibinfo {author} {\bibfnamefont {N.}~\bibnamefont {Davidson}},\ }\href {\doibase 10.1103/PhysRevLett.86.1518} {\bibfield  {journal} {\bibinfo  {journal} {Phys. Rev. Lett.}\ }\textbf {\bibinfo {volume} {86}},\ \bibinfo {pages} {1518} (\bibinfo {year} {2001})}\BibitemShut {NoStop}%
\bibitem [{\citenamefont {Stockmann}\ and\ \citenamefont {Stein}(1990)}]{stockmann90}%
  \BibitemOpen
  \bibfield  {author} {\bibinfo {author} {\bibfnamefont {H.~J.}\ \bibnamefont {Stockmann}}\ and\ \bibinfo {author} {\bibfnamefont {J.}~\bibnamefont {Stein}},\ }\href {\doibase 10.1103/PhysRevLett.64.2215} {\bibfield  {journal} {\bibinfo  {journal} {Phys. Rev. Lett.}\ }\textbf {\bibinfo {volume} {64}},\ \bibinfo {pages} {2215} (\bibinfo {year} {1990})}\BibitemShut {NoStop}%
\bibitem [{\citenamefont {Sridhar}(1991)}]{Sridhar91}%
  \BibitemOpen
  \bibfield  {author} {\bibinfo {author} {\bibfnamefont {S.}~\bibnamefont {Sridhar}},\ }\href {\doibase 10.1103/PhysRevLett.67.785} {\bibfield  {journal} {\bibinfo  {journal} {Phys. Rev. Lett.}\ }\textbf {\bibinfo {volume} {67}},\ \bibinfo {pages} {785} (\bibinfo {year} {1991})}\BibitemShut {NoStop}%
\bibitem [{\citenamefont {Moore}\ \emph {et~al.}(1994)\citenamefont {Moore}, \citenamefont {Robinson}, \citenamefont {Bharucha}, \citenamefont {Williams},\ and\ \citenamefont {Raizen}}]{Moore94}%
  \BibitemOpen
  \bibfield  {author} {\bibinfo {author} {\bibfnamefont {F.~L.}\ \bibnamefont {Moore}}, \bibinfo {author} {\bibfnamefont {J.~C.}\ \bibnamefont {Robinson}}, \bibinfo {author} {\bibfnamefont {C.}~\bibnamefont {Bharucha}}, \bibinfo {author} {\bibfnamefont {P.~E.}\ \bibnamefont {Williams}}, \ and\ \bibinfo {author} {\bibfnamefont {M.~G.}\ \bibnamefont {Raizen}},\ }\href {\doibase 10.1103/PhysRevLett.73.2974} {\bibfield  {journal} {\bibinfo  {journal} {Phys. Rev. Lett.}\ }\textbf {\bibinfo {volume} {73}},\ \bibinfo {pages} {2974} (\bibinfo {year} {1994})}\BibitemShut {NoStop}%
\bibitem [{\citenamefont {Steck}\ \emph {et~al.}(2001)\citenamefont {Steck}, \citenamefont {Oskay},\ and\ \citenamefont {Raizen}}]{Steck01}%
  \BibitemOpen
  \bibfield  {author} {\bibinfo {author} {\bibfnamefont {D.~A.}\ \bibnamefont {Steck}}, \bibinfo {author} {\bibfnamefont {W.~H.}\ \bibnamefont {Oskay}}, \ and\ \bibinfo {author} {\bibfnamefont {M.~G.}\ \bibnamefont {Raizen}},\ }\href {\doibase 10.1126/science.1061569} {\bibfield  {journal} {\bibinfo  {journal} {Science}\ }\textbf {\bibinfo {volume} {293}},\ \bibinfo {pages} {274} (\bibinfo {year} {2001})}\BibitemShut {NoStop}%
\bibitem [{\citenamefont {Hensinger}\ \emph {et~al.}(2001)\citenamefont {Hensinger}, \citenamefont {H\"{a}ffner}, \citenamefont {Browaeys}, \citenamefont {Heckenberg}, \citenamefont {Helmerson}, \citenamefont {McKenzie}, \citenamefont {Milburn}, \citenamefont {Phillips}, \citenamefont {Rolston}, \citenamefont {Rubinsztein-Dunlop},\ and\ \citenamefont {Upcroft}}]{Hensinger01}%
  \BibitemOpen
  \bibfield  {author} {\bibinfo {author} {\bibfnamefont {W.~K.}\ \bibnamefont {Hensinger}}, \bibinfo {author} {\bibfnamefont {H.}~\bibnamefont {H\"{a}ffner}}, \bibinfo {author} {\bibfnamefont {A.}~\bibnamefont {Browaeys}}, \bibinfo {author} {\bibfnamefont {N.~R.}\ \bibnamefont {Heckenberg}}, \bibinfo {author} {\bibfnamefont {K.}~\bibnamefont {Helmerson}}, \bibinfo {author} {\bibfnamefont {C.}~\bibnamefont {McKenzie}}, \bibinfo {author} {\bibfnamefont {G.~J.}\ \bibnamefont {Milburn}}, \bibinfo {author} {\bibfnamefont {W.~D.}\ \bibnamefont {Phillips}}, \bibinfo {author} {\bibfnamefont {S.~L.}\ \bibnamefont {Rolston}}, \bibinfo {author} {\bibfnamefont {H.}~\bibnamefont {Rubinsztein-Dunlop}}, \ and\ \bibinfo {author} {\bibfnamefont {B.}~\bibnamefont {Upcroft}},\ }\href {https://doi.org/10.1038/35083510} {\bibfield  {journal} {\bibinfo  {journal} {Nature}\ }\textbf {\bibinfo {volume} {412}},\ \bibinfo {pages} {52} (\bibinfo {year} {2001})}\BibitemShut {NoStop}%
\bibitem [{\citenamefont {Chaudhury}\ \emph {et~al.}(2009)\citenamefont {Chaudhury}, \citenamefont {Smith}, \citenamefont {Anderson}, \citenamefont {Ghose},\ and\ \citenamefont {Jessen}}]{Chaudhury09}%
  \BibitemOpen
  \bibfield  {author} {\bibinfo {author} {\bibfnamefont {S.}~\bibnamefont {Chaudhury}}, \bibinfo {author} {\bibfnamefont {A.}~\bibnamefont {Smith}}, \bibinfo {author} {\bibfnamefont {B.~E.}\ \bibnamefont {Anderson}}, \bibinfo {author} {\bibfnamefont {S.}~\bibnamefont {Ghose}}, \ and\ \bibinfo {author} {\bibfnamefont {P.~S.}\ \bibnamefont {Jessen}},\ }\href {https://doi.org/10.1038/nature08396} {\bibfield  {journal} {\bibinfo  {journal} {Nature}\ }\textbf {\bibinfo {volume} {461}},\ \bibinfo {pages} {768} (\bibinfo {year} {2009})}\BibitemShut {NoStop}%
\bibitem [{\citenamefont {Weinstein}\ \emph {et~al.}(2002)\citenamefont {Weinstein}, \citenamefont {Lloyd}, \citenamefont {Emerson},\ and\ \citenamefont {Cory}}]{Weinstein02}%
  \BibitemOpen
  \bibfield  {author} {\bibinfo {author} {\bibfnamefont {Y.~S.}\ \bibnamefont {Weinstein}}, \bibinfo {author} {\bibfnamefont {S.}~\bibnamefont {Lloyd}}, \bibinfo {author} {\bibfnamefont {J.}~\bibnamefont {Emerson}}, \ and\ \bibinfo {author} {\bibfnamefont {D.~G.}\ \bibnamefont {Cory}},\ }\href {\doibase 10.1103/PhysRevLett.89.157902} {\bibfield  {journal} {\bibinfo  {journal} {Phys. Rev. Lett.}\ }\textbf {\bibinfo {volume} {89}},\ \bibinfo {pages} {157902} (\bibinfo {year} {2002})}\BibitemShut {NoStop}%
\bibitem [{\citenamefont {Zhang}\ \emph {et~al.}(2022)\citenamefont {Zhang}, \citenamefont {Vidmar},\ and\ \citenamefont {Rigol}}]{Zhang2022}%
  \BibitemOpen
  \bibfield  {author} {\bibinfo {author} {\bibfnamefont {Y.}~\bibnamefont {Zhang}}, \bibinfo {author} {\bibfnamefont {L.}~\bibnamefont {Vidmar}}, \ and\ \bibinfo {author} {\bibfnamefont {M.}~\bibnamefont {Rigol}},\ }\href {\doibase 10.1103/PhysRevE.106.014132} {\bibfield  {journal} {\bibinfo  {journal} {Phys. Rev. E}\ }\textbf {\bibinfo {volume} {106}},\ \bibinfo {pages} {014132} (\bibinfo {year} {2022})}\BibitemShut {NoStop}%
\bibitem [{\citenamefont {\L{}yd\ifmmode~\dot{z}\else \.{z}\fi{}ba}\ \emph {et~al.}(2021{\natexlab{a}})\citenamefont {\L{}yd\ifmmode~\dot{z}\else \.{z}\fi{}ba}, \citenamefont {Zhang}, \citenamefont {Rigol},\ and\ \citenamefont {Vidmar}}]{Rigol202111}%
  \BibitemOpen
  \bibfield  {author} {\bibinfo {author} {\bibfnamefont {P.}~\bibnamefont {\L{}yd\ifmmode~\dot{z}\else \.{z}\fi{}ba}}, \bibinfo {author} {\bibfnamefont {Y.}~\bibnamefont {Zhang}}, \bibinfo {author} {\bibfnamefont {M.}~\bibnamefont {Rigol}}, \ and\ \bibinfo {author} {\bibfnamefont {L.}~\bibnamefont {Vidmar}},\ }\href {\doibase 10.1103/PhysRevB.104.214203} {\bibfield  {journal} {\bibinfo  {journal} {Phys. Rev. B}\ }\textbf {\bibinfo {volume} {104}},\ \bibinfo {pages} {214203} (\bibinfo {year} {2021}{\natexlab{a}})}\BibitemShut {NoStop}%
\bibitem [{\citenamefont {\L{}yd\ifmmode~\dot{z}\else \.{z}\fi{}ba}\ \emph {et~al.}(2021{\natexlab{b}})\citenamefont {\L{}yd\ifmmode~\dot{z}\else \.{z}\fi{}ba}, \citenamefont {Rigol},\ and\ \citenamefont {Vidmar}}]{Vidmar202111}%
  \BibitemOpen
  \bibfield  {author} {\bibinfo {author} {\bibfnamefont {P.}~\bibnamefont {\L{}yd\ifmmode~\dot{z}\else \.{z}\fi{}ba}}, \bibinfo {author} {\bibfnamefont {M.}~\bibnamefont {Rigol}}, \ and\ \bibinfo {author} {\bibfnamefont {L.}~\bibnamefont {Vidmar}},\ }\href {\doibase 10.1103/PhysRevB.103.104206} {\bibfield  {journal} {\bibinfo  {journal} {Phys. Rev. B}\ }\textbf {\bibinfo {volume} {103}},\ \bibinfo {pages} {104206} (\bibinfo {year} {2021}{\natexlab{b}})}\BibitemShut {NoStop}%
\bibitem [{\citenamefont {Santos}\ and\ \citenamefont {Rigol}(2010{\natexlab{a}})}]{Santos2010}%
  \BibitemOpen
  \bibfield  {author} {\bibinfo {author} {\bibfnamefont {L.~F.}\ \bibnamefont {Santos}}\ and\ \bibinfo {author} {\bibfnamefont {M.}~\bibnamefont {Rigol}},\ }\href {\doibase 10.1103/PhysRevE.81.036206} {\bibfield  {journal} {\bibinfo  {journal} {Phys. Rev. E}\ }\textbf {\bibinfo {volume} {81}},\ \bibinfo {pages} {036206} (\bibinfo {year} {2010}{\natexlab{a}})}\BibitemShut {NoStop}%
\bibitem [{\citenamefont {Santos}\ and\ \citenamefont {Rigol}(2010{\natexlab{b}})}]{Santos2010v2}%
  \BibitemOpen
  \bibfield  {author} {\bibinfo {author} {\bibfnamefont {L.~F.}\ \bibnamefont {Santos}}\ and\ \bibinfo {author} {\bibfnamefont {M.}~\bibnamefont {Rigol}},\ }\href {\doibase 10.1103/PhysRevE.82.031130} {\bibfield  {journal} {\bibinfo  {journal} {Phys. Rev. E}\ }\textbf {\bibinfo {volume} {82}},\ \bibinfo {pages} {031130} (\bibinfo {year} {2010}{\natexlab{b}})}\BibitemShut {NoStop}%
\bibitem [{\citenamefont {{Rigol}}(2010)}]{Rigol10}%
  \BibitemOpen
  \bibfield  {author} {\bibinfo {author} {\bibfnamefont {M.}~\bibnamefont {{Rigol}}},\ }\href@noop {} {\bibfield  {journal} {\bibinfo  {journal} {ArXiv e-prints}\ } (\bibinfo {year} {2010})},\ \Eprint {http://arxiv.org/abs/1008.1930} {arXiv:1008.1930 [cond-mat.stat-mech]} \BibitemShut {NoStop}%
\bibitem [{\citenamefont {Kollath}\ \emph {et~al.}(2010)\citenamefont {Kollath}, \citenamefont {Roux}, \citenamefont {Biroli},\ and\ \citenamefont {Läuchli}}]{Kollath_2010}%
  \BibitemOpen
  \bibfield  {author} {\bibinfo {author} {\bibfnamefont {C.}~\bibnamefont {Kollath}}, \bibinfo {author} {\bibfnamefont {G.}~\bibnamefont {Roux}}, \bibinfo {author} {\bibfnamefont {G.}~\bibnamefont {Biroli}}, \ and\ \bibinfo {author} {\bibfnamefont {A.~M.}\ \bibnamefont {Läuchli}},\ }\href {\doibase 10.1088/1742-5468/2010/08/P08011} {\bibfield  {journal} {\bibinfo  {journal} {Journal of Statistical Mechanics: Theory and Experiment}\ }\textbf {\bibinfo {volume} {2010}},\ \bibinfo {pages} {P08011} (\bibinfo {year} {2010})}\BibitemShut {NoStop}%
\bibitem [{\citenamefont {Santos}\ \emph {et~al.}(2012)\citenamefont {Santos}, \citenamefont {Polkovnikov},\ and\ \citenamefont {Rigol}}]{Santos2012}%
  \BibitemOpen
  \bibfield  {author} {\bibinfo {author} {\bibfnamefont {L.~F.}\ \bibnamefont {Santos}}, \bibinfo {author} {\bibfnamefont {A.}~\bibnamefont {Polkovnikov}}, \ and\ \bibinfo {author} {\bibfnamefont {M.}~\bibnamefont {Rigol}},\ }\href {\doibase 10.1103/physreve.86.010102} {\bibfield  {journal} {\bibinfo  {journal} {Physical Review E}\ }\textbf {\bibinfo {volume} {86}} (\bibinfo {year} {2012}),\ 10.1103/physreve.86.010102}\BibitemShut {NoStop}%
\bibitem [{\citenamefont {Richter}\ \emph {et~al.}(2020)\citenamefont {Richter}, \citenamefont {Dymarsky}, \citenamefont {Steinigeweg},\ and\ \citenamefont {Gemmer}}]{Richter_2020}%
  \BibitemOpen
  \bibfield  {author} {\bibinfo {author} {\bibfnamefont {J.}~\bibnamefont {Richter}}, \bibinfo {author} {\bibfnamefont {A.}~\bibnamefont {Dymarsky}}, \bibinfo {author} {\bibfnamefont {R.}~\bibnamefont {Steinigeweg}}, \ and\ \bibinfo {author} {\bibfnamefont {J.}~\bibnamefont {Gemmer}},\ }\href {\doibase 10.1103/physreve.102.042127} {\bibfield  {journal} {\bibinfo  {journal} {Physical Review E}\ }\textbf {\bibinfo {volume} {102}} (\bibinfo {year} {2020}),\ 10.1103/physreve.102.042127}\BibitemShut {NoStop}%
\bibitem [{\citenamefont {Atas}\ \emph {et~al.}(2013{\natexlab{a}})\citenamefont {Atas}, \citenamefont {Bogomolny}, \citenamefont {Giraud},\ and\ \citenamefont {Roux}}]{Atas2013}%
  \BibitemOpen
  \bibfield  {author} {\bibinfo {author} {\bibfnamefont {Y.~Y.}\ \bibnamefont {Atas}}, \bibinfo {author} {\bibfnamefont {E.}~\bibnamefont {Bogomolny}}, \bibinfo {author} {\bibfnamefont {O.}~\bibnamefont {Giraud}}, \ and\ \bibinfo {author} {\bibfnamefont {G.}~\bibnamefont {Roux}},\ }\href {\doibase 10.1103/PhysRevLett.110.084101} {\bibfield  {journal} {\bibinfo  {journal} {Phys. Rev. Lett.}\ }\textbf {\bibinfo {volume} {110}},\ \bibinfo {pages} {084101} (\bibinfo {year} {2013}{\natexlab{a}})}\BibitemShut {NoStop}%
\bibitem [{\citenamefont {Atas}\ \emph {et~al.}(2013{\natexlab{b}})\citenamefont {Atas}, \citenamefont {Bogomolny}, \citenamefont {Giraud}, \citenamefont {Vivo},\ and\ \citenamefont {Vivo}}]{Atas_2013v2}%
  \BibitemOpen
  \bibfield  {author} {\bibinfo {author} {\bibfnamefont {Y.~Y.}\ \bibnamefont {Atas}}, \bibinfo {author} {\bibfnamefont {E.}~\bibnamefont {Bogomolny}}, \bibinfo {author} {\bibfnamefont {O.}~\bibnamefont {Giraud}}, \bibinfo {author} {\bibfnamefont {P.}~\bibnamefont {Vivo}}, \ and\ \bibinfo {author} {\bibfnamefont {E.}~\bibnamefont {Vivo}},\ }\href {\doibase 10.1088/1751-8113/46/35/355204} {\bibfield  {journal} {\bibinfo  {journal} {Journal of Physics A: Mathematical and Theoretical}\ }\textbf {\bibinfo {volume} {46}},\ \bibinfo {pages} {355204} (\bibinfo {year} {2013}{\natexlab{b}})}\BibitemShut {NoStop}%
\bibitem [{\citenamefont {\ifmmode~\check{S}\else \v{S}\fi{}untajs}\ \emph {et~al.}(2020)\citenamefont {\ifmmode~\check{S}\else \v{S}\fi{}untajs}, \citenamefont {Bon\ifmmode~\check{c}\else \v{c}\fi{}a}, \citenamefont {Prosen},\ and\ \citenamefont {Vidmar}}]{Prosen2020}%
  \BibitemOpen
  \bibfield  {author} {\bibinfo {author} {\bibfnamefont {J.}~\bibnamefont {\ifmmode~\check{S}\else \v{S}\fi{}untajs}}, \bibinfo {author} {\bibfnamefont {J.}~\bibnamefont {Bon\ifmmode~\check{c}\else \v{c}\fi{}a}}, \bibinfo {author} {\bibfnamefont {T.~c.~v.}\ \bibnamefont {Prosen}}, \ and\ \bibinfo {author} {\bibfnamefont {L.}~\bibnamefont {Vidmar}},\ }\href {\doibase 10.1103/PhysRevE.102.062144} {\bibfield  {journal} {\bibinfo  {journal} {Phys. Rev. E}\ }\textbf {\bibinfo {volume} {102}},\ \bibinfo {pages} {062144} (\bibinfo {year} {2020})}\BibitemShut {NoStop}%
\bibitem [{\citenamefont {Chan}\ \emph {et~al.}(2018{\natexlab{a}})\citenamefont {Chan}, \citenamefont {De~Luca},\ and\ \citenamefont {Chalker}}]{chan2018solution}%
  \BibitemOpen
  \bibfield  {author} {\bibinfo {author} {\bibfnamefont {A.}~\bibnamefont {Chan}}, \bibinfo {author} {\bibfnamefont {A.}~\bibnamefont {De~Luca}}, \ and\ \bibinfo {author} {\bibfnamefont {J.}~\bibnamefont {Chalker}},\ }\href@noop {} {\bibfield  {journal} {\bibinfo  {journal} {Phys. Rev. X}\ }\textbf {\bibinfo {volume} {8}},\ \bibinfo {pages} {041019} (\bibinfo {year} {2018}{\natexlab{a}})}\BibitemShut {NoStop}%
\bibitem [{\citenamefont {D'Alessio}\ and\ \citenamefont {Rigol}(2014)}]{Luca2014}%
  \BibitemOpen
  \bibfield  {author} {\bibinfo {author} {\bibfnamefont {L.}~\bibnamefont {D'Alessio}}\ and\ \bibinfo {author} {\bibfnamefont {M.}~\bibnamefont {Rigol}},\ }\href {\doibase 10.1103/PhysRevX.4.041048} {\bibfield  {journal} {\bibinfo  {journal} {Phys. Rev. X}\ }\textbf {\bibinfo {volume} {4}},\ \bibinfo {pages} {041048} (\bibinfo {year} {2014})}\BibitemShut {NoStop}%
\bibitem [{\citenamefont {Bertini}\ \emph {et~al.}(2021{\natexlab{a}})\citenamefont {Bertini}, \citenamefont {Kos},\ and\ \citenamefont {Prosen}}]{Bertini2021}%
  \BibitemOpen
  \bibfield  {author} {\bibinfo {author} {\bibfnamefont {B.}~\bibnamefont {Bertini}}, \bibinfo {author} {\bibfnamefont {P.}~\bibnamefont {Kos}}, \ and\ \bibinfo {author} {\bibfnamefont {T.}~\bibnamefont {Prosen}},\ }\href {\doibase 10.1007/s00220-021-04139-2} {\bibfield  {journal} {\bibinfo  {journal} {Communications in Mathematical Physics}\ }\textbf {\bibinfo {volume} {387}},\ \bibinfo {pages} {597} (\bibinfo {year} {2021}{\natexlab{a}})}\BibitemShut {NoStop}%
\bibitem [{\citenamefont {Bertini}\ \emph {et~al.}(2018)\citenamefont {Bertini}, \citenamefont {Kos},\ and\ \citenamefont {Prosen}}]{Bertini2018}%
  \BibitemOpen
  \bibfield  {author} {\bibinfo {author} {\bibfnamefont {B.}~\bibnamefont {Bertini}}, \bibinfo {author} {\bibfnamefont {P.}~\bibnamefont {Kos}}, \ and\ \bibinfo {author} {\bibfnamefont {T.~c.~v.}\ \bibnamefont {Prosen}},\ }\href {\doibase 10.1103/PhysRevLett.121.264101} {\bibfield  {journal} {\bibinfo  {journal} {Phys. Rev. Lett.}\ }\textbf {\bibinfo {volume} {121}},\ \bibinfo {pages} {264101} (\bibinfo {year} {2018})}\BibitemShut {NoStop}%
\bibitem [{\citenamefont {Dyson}(1962{\natexlab{a}})}]{dyson1962I}%
  \BibitemOpen
  \bibfield  {author} {\bibinfo {author} {\bibfnamefont {F.~J.}\ \bibnamefont {Dyson}},\ }\href@noop {} {\bibfield  {journal} {\bibinfo  {journal} {Journal of Mathematical Physics}\ }\textbf {\bibinfo {volume} {3}},\ \bibinfo {pages} {140} (\bibinfo {year} {1962}{\natexlab{a}})}\BibitemShut {NoStop}%
\bibitem [{\citenamefont {Dyson}(1962{\natexlab{b}})}]{dyson1962II}%
  \BibitemOpen
  \bibfield  {author} {\bibinfo {author} {\bibfnamefont {F.~J.}\ \bibnamefont {Dyson}},\ }\href@noop {} {\bibfield  {journal} {\bibinfo  {journal} {Journal of Mathematical Physics}\ }\textbf {\bibinfo {volume} {3}},\ \bibinfo {pages} {157} (\bibinfo {year} {1962}{\natexlab{b}})}\BibitemShut {NoStop}%
\bibitem [{\citenamefont {Dyson}(1962{\natexlab{c}})}]{dyson1962III}%
  \BibitemOpen
  \bibfield  {author} {\bibinfo {author} {\bibfnamefont {F.~J.}\ \bibnamefont {Dyson}},\ }\href@noop {} {\bibfield  {journal} {\bibinfo  {journal} {Journal of Mathematical Physics}\ }\textbf {\bibinfo {volume} {3}},\ \bibinfo {pages} {166} (\bibinfo {year} {1962}{\natexlab{c}})}\BibitemShut {NoStop}%
\bibitem [{\citenamefont {Dyson}\ and\ \citenamefont {Mehta}(1963)}]{dysonmethaIV}%
  \BibitemOpen
  \bibfield  {author} {\bibinfo {author} {\bibfnamefont {F.~J.}\ \bibnamefont {Dyson}}\ and\ \bibinfo {author} {\bibfnamefont {M.~L.}\ \bibnamefont {Mehta}},\ }\href@noop {} {\bibfield  {journal} {\bibinfo  {journal} {Journal of Mathematical Physics}\ }\textbf {\bibinfo {volume} {4}},\ \bibinfo {pages} {701} (\bibinfo {year} {1963})}\BibitemShut {NoStop}%
\bibitem [{\citenamefont {Mehta}\ and\ \citenamefont {Dyson}(1963)}]{dysonmehta1963V}%
  \BibitemOpen
  \bibfield  {author} {\bibinfo {author} {\bibfnamefont {M.}~\bibnamefont {Mehta}}\ and\ \bibinfo {author} {\bibfnamefont {F.}~\bibnamefont {Dyson}},\ }\href@noop {} {\bibfield  {journal} {\bibinfo  {journal} {Journal of Mathematical Physics (New York)(US)}\ }\textbf {\bibinfo {volume} {4}} (\bibinfo {year} {1963})}\BibitemShut {NoStop}%
\bibitem [{\citenamefont {Dyson}(1962{\natexlab{d}})}]{dyson1962threefold}%
  \BibitemOpen
  \bibfield  {author} {\bibinfo {author} {\bibfnamefont {F.~J.}\ \bibnamefont {Dyson}},\ }\href@noop {} {\bibfield  {journal} {\bibinfo  {journal} {Journal of Mathematical Physics}\ }\textbf {\bibinfo {volume} {3}},\ \bibinfo {pages} {1199} (\bibinfo {year} {1962}{\natexlab{d}})}\BibitemShut {NoStop}%
\bibitem [{\citenamefont {Bohigas}\ \emph {et~al.}(1984)\citenamefont {Bohigas}, \citenamefont {Giannoni},\ and\ \citenamefont {Schmit}}]{bohigas1984characterization}%
  \BibitemOpen
  \bibfield  {author} {\bibinfo {author} {\bibfnamefont {O.}~\bibnamefont {Bohigas}}, \bibinfo {author} {\bibfnamefont {M.-J.}\ \bibnamefont {Giannoni}}, \ and\ \bibinfo {author} {\bibfnamefont {C.}~\bibnamefont {Schmit}},\ }\href@noop {} {\bibfield  {journal} {\bibinfo  {journal} {Physical review letters}\ }\textbf {\bibinfo {volume} {52}},\ \bibinfo {pages} {1} (\bibinfo {year} {1984})}\BibitemShut {NoStop}%
\bibitem [{\citenamefont {Bohrdt}\ \emph {et~al.}(2017)\citenamefont {Bohrdt}, \citenamefont {Mendl}, \citenamefont {Endres},\ and\ \citenamefont {Knap}}]{bohrdt2017scrambling}%
  \BibitemOpen
  \bibfield  {author} {\bibinfo {author} {\bibfnamefont {A.}~\bibnamefont {Bohrdt}}, \bibinfo {author} {\bibfnamefont {C.~B.}\ \bibnamefont {Mendl}}, \bibinfo {author} {\bibfnamefont {M.}~\bibnamefont {Endres}}, \ and\ \bibinfo {author} {\bibfnamefont {M.}~\bibnamefont {Knap}},\ }\href {\doibase 10.1088/1367-2630/aa719b} {\bibfield  {journal} {\bibinfo  {journal} {New. J. Phys.}\ }\textbf {\bibinfo {volume} {19}},\ \bibinfo {pages} {063001} (\bibinfo {year} {2017})}\BibitemShut {NoStop}%
\bibitem [{\citenamefont {Andreev}\ \emph {et~al.}(1996)\citenamefont {Andreev}, \citenamefont {Agam}, \citenamefont {Simons},\ and\ \citenamefont {Altshuler}}]{andreev1996quantum}%
  \BibitemOpen
  \bibfield  {author} {\bibinfo {author} {\bibfnamefont {A.}~\bibnamefont {Andreev}}, \bibinfo {author} {\bibfnamefont {O.}~\bibnamefont {Agam}}, \bibinfo {author} {\bibfnamefont {B.}~\bibnamefont {Simons}}, \ and\ \bibinfo {author} {\bibfnamefont {B.}~\bibnamefont {Altshuler}},\ }\href@noop {} {\bibfield  {journal} {\bibinfo  {journal} {Physical review letters}\ }\textbf {\bibinfo {volume} {76}},\ \bibinfo {pages} {3947} (\bibinfo {year} {1996})}\BibitemShut {NoStop}%
\bibitem [{\citenamefont {M{\"u}ller}\ \emph {et~al.}(2004)\citenamefont {M{\"u}ller}, \citenamefont {Heusler}, \citenamefont {Braun}, \citenamefont {Haake},\ and\ \citenamefont {Altland}}]{muller2004semiclassical}%
  \BibitemOpen
  \bibfield  {author} {\bibinfo {author} {\bibfnamefont {S.}~\bibnamefont {M{\"u}ller}}, \bibinfo {author} {\bibfnamefont {S.}~\bibnamefont {Heusler}}, \bibinfo {author} {\bibfnamefont {P.}~\bibnamefont {Braun}}, \bibinfo {author} {\bibfnamefont {F.}~\bibnamefont {Haake}}, \ and\ \bibinfo {author} {\bibfnamefont {A.}~\bibnamefont {Altland}},\ }\href@noop {} {\bibfield  {journal} {\bibinfo  {journal} {Physical review letters}\ }\textbf {\bibinfo {volume} {93}},\ \bibinfo {pages} {014103} (\bibinfo {year} {2004})}\BibitemShut {NoStop}%
\bibitem [{\citenamefont {Mehta}(2004)}]{mehta2004random}%
  \BibitemOpen
  \bibfield  {author} {\bibinfo {author} {\bibfnamefont {M.~L.}\ \bibnamefont {Mehta}},\ }\href@noop {} {\emph {\bibinfo {title} {Random matrices}}}\ (\bibinfo  {publisher} {Elsevier},\ \bibinfo {year} {2004})\BibitemShut {NoStop}%
\bibitem [{\citenamefont {Bruus}\ and\ \citenamefont {Angl`es~d'Auriac}(1997)}]{Bruus1997}%
  \BibitemOpen
  \bibfield  {author} {\bibinfo {author} {\bibfnamefont {H.}~\bibnamefont {Bruus}}\ and\ \bibinfo {author} {\bibfnamefont {J.-C.}\ \bibnamefont {Angl`es~d'Auriac}},\ }\href {\doibase 10.1103/PhysRevB.55.9142} {\bibfield  {journal} {\bibinfo  {journal} {Phys. Rev. B}\ }\textbf {\bibinfo {volume} {55}},\ \bibinfo {pages} {9142} (\bibinfo {year} {1997})}\BibitemShut {NoStop}%
\bibitem [{\citenamefont {Mehta}(1960)}]{mehta1960statistical}%
  \BibitemOpen
  \bibfield  {author} {\bibinfo {author} {\bibfnamefont {M.~L.}\ \bibnamefont {Mehta}},\ }\href@noop {} {\bibfield  {journal} {\bibinfo  {journal} {Nuclear Physics}\ }\textbf {\bibinfo {volume} {18}},\ \bibinfo {pages} {395} (\bibinfo {year} {1960})}\BibitemShut {NoStop}%
\bibitem [{\citenamefont {Jimbo}\ \emph {et~al.}(1980)\citenamefont {Jimbo}, \citenamefont {Miwa}, \citenamefont {Mori},\ and\ \citenamefont {Sato}}]{Jimbo1980DensityMO}%
  \BibitemOpen
  \bibfield  {author} {\bibinfo {author} {\bibfnamefont {M.}~\bibnamefont {Jimbo}}, \bibinfo {author} {\bibfnamefont {T.}~\bibnamefont {Miwa}}, \bibinfo {author} {\bibfnamefont {Y.}~\bibnamefont {Mori}}, \ and\ \bibinfo {author} {\bibfnamefont {M.}~\bibnamefont {Sato}},\ }\href@noop {} {\bibfield  {journal} {\bibinfo  {journal} {Physica D: Nonlinear Phenomena}\ }\textbf {\bibinfo {volume} {1}},\ \bibinfo {pages} {80} (\bibinfo {year} {1980})}\BibitemShut {NoStop}%
\bibitem [{\citenamefont {Forrester}\ and\ \citenamefont {Witte}(2000)}]{forrester2000exact}%
  \BibitemOpen
  \bibfield  {author} {\bibinfo {author} {\bibfnamefont {P.}~\bibnamefont {Forrester}}\ and\ \bibinfo {author} {\bibfnamefont {N.}~\bibnamefont {Witte}},\ }\href@noop {} {\bibfield  {journal} {\bibinfo  {journal} {Letters in Mathematical Physics}\ }\textbf {\bibinfo {volume} {53}},\ \bibinfo {pages} {195} (\bibinfo {year} {2000})}\BibitemShut {NoStop}%
\bibitem [{\citenamefont {Imbrie}(2016)}]{imbrie2016many}%
  \BibitemOpen
  \bibfield  {author} {\bibinfo {author} {\bibfnamefont {J.~Z.}\ \bibnamefont {Imbrie}},\ }\href@noop {} {\bibfield  {journal} {\bibinfo  {journal} {Journal of Statistical Physics}\ }\textbf {\bibinfo {volume} {163}},\ \bibinfo {pages} {998} (\bibinfo {year} {2016})}\BibitemShut {NoStop}%
\bibitem [{\citenamefont {Livan}\ \emph {et~al.}(2018)\citenamefont {Livan}, \citenamefont {Novaes},\ and\ \citenamefont {Vivo}}]{livan2018introduction}%
  \BibitemOpen
  \bibfield  {author} {\bibinfo {author} {\bibfnamefont {G.}~\bibnamefont {Livan}}, \bibinfo {author} {\bibfnamefont {M.}~\bibnamefont {Novaes}}, \ and\ \bibinfo {author} {\bibfnamefont {P.}~\bibnamefont {Vivo}},\ }\href@noop {} {\bibfield  {journal} {\bibinfo  {journal} {Monograph Award}\ }\textbf {\bibinfo {volume} {63}} (\bibinfo {year} {2018})}\BibitemShut {NoStop}%
\bibitem [{\citenamefont {Alhambra}\ \emph {et~al.}(2020)\citenamefont {Alhambra}, \citenamefont {Riddell},\ and\ \citenamefont {Garc{\'\i}a-Pintos}}]{alhambra2020time}%
  \BibitemOpen
  \bibfield  {author} {\bibinfo {author} {\bibfnamefont {{\'A}.~M.}\ \bibnamefont {Alhambra}}, \bibinfo {author} {\bibfnamefont {J.}~\bibnamefont {Riddell}}, \ and\ \bibinfo {author} {\bibfnamefont {L.~P.}\ \bibnamefont {Garc{\'\i}a-Pintos}},\ }\href {\doibase 10.1103/PhysRevLett.124.110605} {\bibfield  {journal} {\bibinfo  {journal} {Phys. Rev. Lett.}\ }\textbf {\bibinfo {volume} {124}},\ \bibinfo {pages} {110605} (\bibinfo {year} {2020})}\BibitemShut {NoStop}%
\bibitem [{\citenamefont {Riddell}\ and\ \citenamefont {S\o{}rensen}(2020)}]{Riddell2020}%
  \BibitemOpen
  \bibfield  {author} {\bibinfo {author} {\bibfnamefont {J.}~\bibnamefont {Riddell}}\ and\ \bibinfo {author} {\bibfnamefont {E.~S.}\ \bibnamefont {S\o{}rensen}},\ }\href {\doibase 10.1103/PhysRevB.101.024202} {\bibfield  {journal} {\bibinfo  {journal} {Phys. Rev. B}\ }\textbf {\bibinfo {volume} {101}},\ \bibinfo {pages} {024202} (\bibinfo {year} {2020})}\BibitemShut {NoStop}%
\bibitem [{\citenamefont {Gogolin}\ and\ \citenamefont {Eisert}(2016)}]{gogolin2016equilibration}%
  \BibitemOpen
  \bibfield  {author} {\bibinfo {author} {\bibfnamefont {C.}~\bibnamefont {Gogolin}}\ and\ \bibinfo {author} {\bibfnamefont {J.}~\bibnamefont {Eisert}},\ }\href {\doibase 10.1088/0034-4885/79/5/056001} {\bibfield  {journal} {\bibinfo  {journal} {Rep. Prog. Phys.}\ }\textbf {\bibinfo {volume} {79}},\ \bibinfo {pages} {056001} (\bibinfo {year} {2016})}\BibitemShut {NoStop}%
\bibitem [{\citenamefont {Masanes}\ \emph {et~al.}(2013)\citenamefont {Masanes}, \citenamefont {Roncaglia},\ and\ \citenamefont {Ac\'{\i}n}}]{masanes2013complexity}%
  \BibitemOpen
  \bibfield  {author} {\bibinfo {author} {\bibfnamefont {L.}~\bibnamefont {Masanes}}, \bibinfo {author} {\bibfnamefont {A.~J.}\ \bibnamefont {Roncaglia}}, \ and\ \bibinfo {author} {\bibfnamefont {A.}~\bibnamefont {Ac\'{\i}n}},\ }\href {\doibase 10.1103/PhysRevE.87.032137} {\bibfield  {journal} {\bibinfo  {journal} {Phys. Rev. E}\ }\textbf {\bibinfo {volume} {87}},\ \bibinfo {pages} {032137} (\bibinfo {year} {2013})}\BibitemShut {NoStop}%
\bibitem [{\citenamefont {Wilming}\ \emph {et~al.}(2018)\citenamefont {Wilming}, \citenamefont {de~Oliveira}, \citenamefont {Short},\ and\ \citenamefont {Eisert}}]{wilming2018equilibration}%
  \BibitemOpen
  \bibfield  {author} {\bibinfo {author} {\bibfnamefont {H.}~\bibnamefont {Wilming}}, \bibinfo {author} {\bibfnamefont {T.~R.}\ \bibnamefont {de~Oliveira}}, \bibinfo {author} {\bibfnamefont {A.~J.}\ \bibnamefont {Short}}, \ and\ \bibinfo {author} {\bibfnamefont {J.}~\bibnamefont {Eisert}},\ }in\ \href@noop {} {\emph {\bibinfo {booktitle} {Thermodynamics in the Quantum Regime}}}\ (\bibinfo  {publisher} {Springer},\ \bibinfo {year} {2018})\ pp.\ \bibinfo {pages} {435--455}\BibitemShut {NoStop}%
\bibitem [{\citenamefont {Heveling}\ \emph {et~al.}(2020)\citenamefont {Heveling}, \citenamefont {Knipschild},\ and\ \citenamefont {Gemmer}}]{heveling2020compelling}%
  \BibitemOpen
  \bibfield  {author} {\bibinfo {author} {\bibfnamefont {R.}~\bibnamefont {Heveling}}, \bibinfo {author} {\bibfnamefont {L.}~\bibnamefont {Knipschild}}, \ and\ \bibinfo {author} {\bibfnamefont {J.}~\bibnamefont {Gemmer}},\ }\href {\doibase 10.1088/1751-8121/ab9e2b} {\bibfield  {journal} {\bibinfo  {journal} {Journal of Physics A: Mathematical and Theoretical}\ }\textbf {\bibinfo {volume} {53}},\ \bibinfo {pages} {375303} (\bibinfo {year} {2020})}\BibitemShut {NoStop}%
\bibitem [{\citenamefont {Campos~Venuti}\ and\ \citenamefont {Zanardi}(2010)}]{Campos_Venuti_2010}%
  \BibitemOpen
  \bibfield  {author} {\bibinfo {author} {\bibfnamefont {L.}~\bibnamefont {Campos~Venuti}}\ and\ \bibinfo {author} {\bibfnamefont {P.}~\bibnamefont {Zanardi}},\ }\href {\doibase 10.1103/physreva.81.022113} {\bibfield  {journal} {\bibinfo  {journal} {Phys. Rev. A}\ }\textbf {\bibinfo {volume} {81}} (\bibinfo {year} {2010}),\ 10.1103/physreva.81.022113}\BibitemShut {NoStop}%
\bibitem [{\citenamefont {Knipschild}\ and\ \citenamefont {Gemmer}(2020)}]{Knipschild2020}%
  \BibitemOpen
  \bibfield  {author} {\bibinfo {author} {\bibfnamefont {L.}~\bibnamefont {Knipschild}}\ and\ \bibinfo {author} {\bibfnamefont {J.}~\bibnamefont {Gemmer}},\ }\href {\doibase 10.1103/physreve.101.062205} {\bibfield  {journal} {\bibinfo  {journal} {Physical Review E}\ }\textbf {\bibinfo {volume} {101}} (\bibinfo {year} {2020}),\ 10.1103/physreve.101.062205}\BibitemShut {NoStop}%
\bibitem [{\citenamefont {Carvalho}\ \emph {et~al.}(2023)\citenamefont {Carvalho}, \citenamefont {dos Prazeres}, \citenamefont {Correia},\ and\ \citenamefont {de~Oliveira}}]{carvalho2023equilibration}%
  \BibitemOpen
  \bibfield  {author} {\bibinfo {author} {\bibfnamefont {G.~D.}\ \bibnamefont {Carvalho}}, \bibinfo {author} {\bibfnamefont {L.~F.}\ \bibnamefont {dos Prazeres}}, \bibinfo {author} {\bibfnamefont {P.~S.}\ \bibnamefont {Correia}}, \ and\ \bibinfo {author} {\bibfnamefont {T.~R.}\ \bibnamefont {de~Oliveira}},\ }\href@noop {} {\  (\bibinfo {year} {2023})},\ \Eprint {http://arxiv.org/abs/2305.11985} {arXiv:2305.11985 [quant-ph]} \BibitemShut {NoStop}%
\bibitem [{\citenamefont {Riddell}\ \emph {et~al.}(2021)\citenamefont {Riddell}, \citenamefont {Kirkby}, \citenamefont {O'Dell},\ and\ \citenamefont {Sørensen}}]{riddell2021scaling}%
  \BibitemOpen
  \bibfield  {author} {\bibinfo {author} {\bibfnamefont {J.}~\bibnamefont {Riddell}}, \bibinfo {author} {\bibfnamefont {W.}~\bibnamefont {Kirkby}}, \bibinfo {author} {\bibfnamefont {D.~H.~J.}\ \bibnamefont {O'Dell}}, \ and\ \bibinfo {author} {\bibfnamefont {E.~S.}\ \bibnamefont {Sørensen}},\ }\href@noop {} {\enquote {\bibinfo {title} {Scaling at the otoc wavefront: Integrable versus chaotic models},}\ } (\bibinfo {year} {2021}),\ \Eprint {http://arxiv.org/abs/2111.01336} {arXiv:2111.01336 [cond-mat.stat-mech]} \BibitemShut {NoStop}%
\bibitem [{\citenamefont {Riddell}\ and\ \citenamefont {S{\o}rensen}(2019)}]{riddell2019out}%
  \BibitemOpen
  \bibfield  {author} {\bibinfo {author} {\bibfnamefont {J.}~\bibnamefont {Riddell}}\ and\ \bibinfo {author} {\bibfnamefont {E.~S.}\ \bibnamefont {S{\o}rensen}},\ }\href@noop {} {\bibfield  {journal} {\bibinfo  {journal} {Physical Review B}\ }\textbf {\bibinfo {volume} {99}},\ \bibinfo {pages} {054205} (\bibinfo {year} {2019})}\BibitemShut {NoStop}%
\bibitem [{\citenamefont {Yoshida}\ and\ \citenamefont {Yao}(2019)}]{Yoshida2019}%
  \BibitemOpen
  \bibfield  {author} {\bibinfo {author} {\bibfnamefont {B.}~\bibnamefont {Yoshida}}\ and\ \bibinfo {author} {\bibfnamefont {N.~Y.}\ \bibnamefont {Yao}},\ }\href {\doibase 10.1103/PhysRevX.9.011006} {\bibfield  {journal} {\bibinfo  {journal} {Phys. Rev. X}\ }\textbf {\bibinfo {volume} {9}},\ \bibinfo {pages} {011006} (\bibinfo {year} {2019})}\BibitemShut {NoStop}%
\bibitem [{\citenamefont {Fortes}\ \emph {et~al.}(2019)\citenamefont {Fortes}, \citenamefont {Garc\'{\i}a-Mata}, \citenamefont {Jalabert},\ and\ \citenamefont {Wisniacki}}]{Fortes2019}%
  \BibitemOpen
  \bibfield  {author} {\bibinfo {author} {\bibfnamefont {E.~M.}\ \bibnamefont {Fortes}}, \bibinfo {author} {\bibfnamefont {I.}~\bibnamefont {Garc\'{\i}a-Mata}}, \bibinfo {author} {\bibfnamefont {R.~A.}\ \bibnamefont {Jalabert}}, \ and\ \bibinfo {author} {\bibfnamefont {D.~A.}\ \bibnamefont {Wisniacki}},\ }\href {\doibase 10.1103/PhysRevE.100.042201} {\bibfield  {journal} {\bibinfo  {journal} {Phys. Rev. E}\ }\textbf {\bibinfo {volume} {100}},\ \bibinfo {pages} {042201} (\bibinfo {year} {2019})}\BibitemShut {NoStop}%
\bibitem [{\citenamefont {Shukla}\ \emph {et~al.}(2022)\citenamefont {Shukla}, \citenamefont {Lakshminarayan},\ and\ \citenamefont {Mishra}}]{Shukla2022}%
  \BibitemOpen
  \bibfield  {author} {\bibinfo {author} {\bibfnamefont {R.~K.}\ \bibnamefont {Shukla}}, \bibinfo {author} {\bibfnamefont {A.}~\bibnamefont {Lakshminarayan}}, \ and\ \bibinfo {author} {\bibfnamefont {S.~K.}\ \bibnamefont {Mishra}},\ }\href {\doibase 10.1103/PhysRevB.105.224307} {\bibfield  {journal} {\bibinfo  {journal} {Phys. Rev. B}\ }\textbf {\bibinfo {volume} {105}},\ \bibinfo {pages} {224307} (\bibinfo {year} {2022})}\BibitemShut {NoStop}%
\bibitem [{\citenamefont {Fortes}\ \emph {et~al.}(2020)\citenamefont {Fortes}, \citenamefont {García-Mata}, \citenamefont {Jalabert},\ and\ \citenamefont {Wisniacki}}]{Fortes_2020v2}%
  \BibitemOpen
  \bibfield  {author} {\bibinfo {author} {\bibfnamefont {E.~M.}\ \bibnamefont {Fortes}}, \bibinfo {author} {\bibfnamefont {I.}~\bibnamefont {García-Mata}}, \bibinfo {author} {\bibfnamefont {R.~A.}\ \bibnamefont {Jalabert}}, \ and\ \bibinfo {author} {\bibfnamefont {D.~A.}\ \bibnamefont {Wisniacki}},\ }\href {\doibase 10.1209/0295-5075/130/60001} {\bibfield  {journal} {\bibinfo  {journal} {Europhysics Letters}\ }\textbf {\bibinfo {volume} {130}},\ \bibinfo {pages} {60001} (\bibinfo {year} {2020})}\BibitemShut {NoStop}%
\bibitem [{\citenamefont {Mark}\ \emph {et~al.}(2022)\citenamefont {Mark}, \citenamefont {Choi}, \citenamefont {Shaw}, \citenamefont {Endres},\ and\ \citenamefont {Choi}}]{Mark2022}%
  \BibitemOpen
  \bibfield  {author} {\bibinfo {author} {\bibfnamefont {D.~K.}\ \bibnamefont {Mark}}, \bibinfo {author} {\bibfnamefont {J.}~\bibnamefont {Choi}}, \bibinfo {author} {\bibfnamefont {A.~L.}\ \bibnamefont {Shaw}}, \bibinfo {author} {\bibfnamefont {M.}~\bibnamefont {Endres}}, \ and\ \bibinfo {author} {\bibfnamefont {S.}~\bibnamefont {Choi}},\ }\href {\doibase 10.48550/ARXIV.2205.12211} {\enquote {\bibinfo {title} {Benchmarking quantum simulators using quantum chaos},}\ } (\bibinfo {year} {2022})\BibitemShut {NoStop}%
\bibitem [{\citenamefont {Srednicki}(1999)}]{Srednicki99}%
  \BibitemOpen
  \bibfield  {author} {\bibinfo {author} {\bibfnamefont {M.}~\bibnamefont {Srednicki}},\ }\href {\doibase 10.1088/0305-4470/32/7/007} {\bibfield  {journal} {\bibinfo  {journal} {Journal of Physics A: Mathematical and General}\ }\textbf {\bibinfo {volume} {32}},\ \bibinfo {pages} {1163} (\bibinfo {year} {1999})}\BibitemShut {NoStop}%
\bibitem [{\citenamefont {Riddell}\ \emph {et~al.}(2023)\citenamefont {Riddell}, \citenamefont {Pagliaroli},\ and\ \citenamefont {Álvaro M.~Alhambra}}]{Riddell2022}%
  \BibitemOpen
  \bibfield  {author} {\bibinfo {author} {\bibfnamefont {J.}~\bibnamefont {Riddell}}, \bibinfo {author} {\bibfnamefont {N.~J.}\ \bibnamefont {Pagliaroli}}, \ and\ \bibinfo {author} {\bibnamefont {Álvaro M.~Alhambra}},\ }\href {\doibase 10.21468/SciPostPhys.15.4.165} {\bibfield  {journal} {\bibinfo  {journal} {SciPost Phys.}\ }\textbf {\bibinfo {volume} {15}},\ \bibinfo {pages} {165} (\bibinfo {year} {2023})}\BibitemShut {NoStop}%
\bibitem [{\citenamefont {Khalkhali}\ and\ \citenamefont {Pagliaroli}(2022)}]{khalkhali2022spectral}%
  \BibitemOpen
  \bibfield  {author} {\bibinfo {author} {\bibfnamefont {M.}~\bibnamefont {Khalkhali}}\ and\ \bibinfo {author} {\bibfnamefont {N.}~\bibnamefont {Pagliaroli}},\ }\href@noop {} {\bibfield  {journal} {\bibinfo  {journal} {Journal of Mathematical Physics}\ }\textbf {\bibinfo {volume} {63}},\ \bibinfo {pages} {053504} (\bibinfo {year} {2022})}\BibitemShut {NoStop}%
\bibitem [{\citenamefont {Srivastava}\ \emph {et~al.}(2016)\citenamefont {Srivastava}, \citenamefont {Tomsovic}, \citenamefont {Lakshminarayan}, \citenamefont {Ketzmerick},\ and\ \citenamefont {B\"acker}}]{Srivastava2016}%
  \BibitemOpen
  \bibfield  {author} {\bibinfo {author} {\bibfnamefont {S.~C.~L.}\ \bibnamefont {Srivastava}}, \bibinfo {author} {\bibfnamefont {S.}~\bibnamefont {Tomsovic}}, \bibinfo {author} {\bibfnamefont {A.}~\bibnamefont {Lakshminarayan}}, \bibinfo {author} {\bibfnamefont {R.}~\bibnamefont {Ketzmerick}}, \ and\ \bibinfo {author} {\bibfnamefont {A.}~\bibnamefont {B\"acker}},\ }\href {\doibase 10.1103/PhysRevLett.116.054101} {\bibfield  {journal} {\bibinfo  {journal} {Phys. Rev. Lett.}\ }\textbf {\bibinfo {volume} {116}},\ \bibinfo {pages} {054101} (\bibinfo {year} {2016})}\BibitemShut {NoStop}%
\bibitem [{\citenamefont {Herrmann}\ \emph {et~al.}(2020)\citenamefont {Herrmann}, \citenamefont {Kieler}, \citenamefont {Fritzsch},\ and\ \citenamefont {B\"acker}}]{Herrmann}%
  \BibitemOpen
  \bibfield  {author} {\bibinfo {author} {\bibfnamefont {T.}~\bibnamefont {Herrmann}}, \bibinfo {author} {\bibfnamefont {M.~F.~I.}\ \bibnamefont {Kieler}}, \bibinfo {author} {\bibfnamefont {F.}~\bibnamefont {Fritzsch}}, \ and\ \bibinfo {author} {\bibfnamefont {A.}~\bibnamefont {B\"acker}},\ }\href {\doibase 10.1103/PhysRevE.101.022221} {\bibfield  {journal} {\bibinfo  {journal} {Phys. Rev. E}\ }\textbf {\bibinfo {volume} {101}},\ \bibinfo {pages} {022221} (\bibinfo {year} {2020})}\BibitemShut {NoStop}%
\bibitem [{\citenamefont {LeBlond}\ \emph {et~al.}(2019)\citenamefont {LeBlond}, \citenamefont {Mallayya}, \citenamefont {Vidmar},\ and\ \citenamefont {Rigol}}]{LeBlond_2019}%
  \BibitemOpen
  \bibfield  {author} {\bibinfo {author} {\bibfnamefont {T.}~\bibnamefont {LeBlond}}, \bibinfo {author} {\bibfnamefont {K.}~\bibnamefont {Mallayya}}, \bibinfo {author} {\bibfnamefont {L.}~\bibnamefont {Vidmar}}, \ and\ \bibinfo {author} {\bibfnamefont {M.}~\bibnamefont {Rigol}},\ }\href {\doibase 10.1103/PhysRevE.100.062134} {\bibfield  {journal} {\bibinfo  {journal} {Phys. Rev. E}\ }\textbf {\bibinfo {volume} {100}},\ \bibinfo {pages} {062134} (\bibinfo {year} {2019})}\BibitemShut {NoStop}%
\bibitem [{\citenamefont {\L{}yd\ifmmode~\dot{z}\else \.{z}\fi{}ba}\ \emph {et~al.}(2021{\natexlab{c}})\citenamefont {\L{}yd\ifmmode~\dot{z}\else \.{z}\fi{}ba}, \citenamefont {Rigol},\ and\ \citenamefont {Vidmar}}]{Rigol2021}%
  \BibitemOpen
  \bibfield  {author} {\bibinfo {author} {\bibfnamefont {P.}~\bibnamefont {\L{}yd\ifmmode~\dot{z}\else \.{z}\fi{}ba}}, \bibinfo {author} {\bibfnamefont {M.}~\bibnamefont {Rigol}}, \ and\ \bibinfo {author} {\bibfnamefont {L.}~\bibnamefont {Vidmar}},\ }\href {\doibase 10.1103/PhysRevB.103.104206} {\bibfield  {journal} {\bibinfo  {journal} {Phys. Rev. B}\ }\textbf {\bibinfo {volume} {103}},\ \bibinfo {pages} {104206} (\bibinfo {year} {2021}{\natexlab{c}})}\BibitemShut {NoStop}%
\bibitem [{\citenamefont {Oganesyan}\ and\ \citenamefont {Huse}(2007)}]{oganesyan2007localization}%
  \BibitemOpen
  \bibfield  {author} {\bibinfo {author} {\bibfnamefont {V.}~\bibnamefont {Oganesyan}}\ and\ \bibinfo {author} {\bibfnamefont {D.~A.}\ \bibnamefont {Huse}},\ }\href@noop {} {\bibfield  {journal} {\bibinfo  {journal} {Physical review b}\ }\textbf {\bibinfo {volume} {75}},\ \bibinfo {pages} {155111} (\bibinfo {year} {2007})}\BibitemShut {NoStop}%
\bibitem [{\citenamefont {Atas}\ \emph {et~al.}(2013{\natexlab{c}})\citenamefont {Atas}, \citenamefont {Bogomolny}, \citenamefont {Giraud},\ and\ \citenamefont {Roux}}]{atas2013distribution}%
  \BibitemOpen
  \bibfield  {author} {\bibinfo {author} {\bibfnamefont {Y.}~\bibnamefont {Atas}}, \bibinfo {author} {\bibfnamefont {E.}~\bibnamefont {Bogomolny}}, \bibinfo {author} {\bibfnamefont {O.}~\bibnamefont {Giraud}}, \ and\ \bibinfo {author} {\bibfnamefont {G.}~\bibnamefont {Roux}},\ }\href@noop {} {\bibfield  {journal} {\bibinfo  {journal} {Physical review letters}\ }\textbf {\bibinfo {volume} {110}},\ \bibinfo {pages} {084101} (\bibinfo {year} {2013}{\natexlab{c}})}\BibitemShut {NoStop}%
\bibitem [{\citenamefont {Short}(2011)}]{short2011equilibration}%
  \BibitemOpen
  \bibfield  {author} {\bibinfo {author} {\bibfnamefont {A.~J.}\ \bibnamefont {Short}},\ }\href {\doibase 10.1088/1367-2630/13/5/053009} {\bibfield  {journal} {\bibinfo  {journal} {New. J. Phys.}\ }\textbf {\bibinfo {volume} {13}},\ \bibinfo {pages} {053009} (\bibinfo {year} {2011})}\BibitemShut {NoStop}%
\bibitem [{\citenamefont {Bertini}\ \emph {et~al.}(2021{\natexlab{b}})\citenamefont {Bertini}, \citenamefont {Kos},\ and\ \citenamefont {Prosen}}]{bertini2021random}%
  \BibitemOpen
  \bibfield  {author} {\bibinfo {author} {\bibfnamefont {B.}~\bibnamefont {Bertini}}, \bibinfo {author} {\bibfnamefont {P.}~\bibnamefont {Kos}}, \ and\ \bibinfo {author} {\bibfnamefont {T.}~\bibnamefont {Prosen}},\ }\href@noop {} {\bibfield  {journal} {\bibinfo  {journal} {Communications in Mathematical Physics}\ }\textbf {\bibinfo {volume} {387}},\ \bibinfo {pages} {597} (\bibinfo {year} {2021}{\natexlab{b}})}\BibitemShut {NoStop}%
\bibitem [{\citenamefont {Chan}\ \emph {et~al.}(2018{\natexlab{b}})\citenamefont {Chan}, \citenamefont {De~Luca},\ and\ \citenamefont {Chalker}}]{chan2018}%
  \BibitemOpen
  \bibfield  {author} {\bibinfo {author} {\bibfnamefont {A.}~\bibnamefont {Chan}}, \bibinfo {author} {\bibfnamefont {A.}~\bibnamefont {De~Luca}}, \ and\ \bibinfo {author} {\bibfnamefont {J.~T.}\ \bibnamefont {Chalker}},\ }\href {\doibase 10.1103/PhysRevX.8.041019} {\bibfield  {journal} {\bibinfo  {journal} {Phys. Rev. X}\ }\textbf {\bibinfo {volume} {8}},\ \bibinfo {pages} {041019} (\bibinfo {year} {2018}{\natexlab{b}})}\BibitemShut {NoStop}%
\bibitem [{\citenamefont {Riddell}\ and\ \citenamefont {Bertini}(2024)}]{riddell2024generic}%
  \BibitemOpen
  \bibfield  {author} {\bibinfo {author} {\bibfnamefont {J.}~\bibnamefont {Riddell}}\ and\ \bibinfo {author} {\bibfnamefont {B.}~\bibnamefont {Bertini}},\ }\href@noop {} {\bibfield  {journal} {\bibinfo  {journal} {arXiv preprint arXiv:2404.12100}\ } (\bibinfo {year} {2024})}\BibitemShut {NoStop}%
\bibitem [{\citenamefont {Reimann}\ and\ \citenamefont {Kastner}(2012)}]{reimann2012equilibration}%
  \BibitemOpen
  \bibfield  {author} {\bibinfo {author} {\bibfnamefont {P.}~\bibnamefont {Reimann}}\ and\ \bibinfo {author} {\bibfnamefont {M.}~\bibnamefont {Kastner}},\ }\href@noop {} {\bibfield  {journal} {\bibinfo  {journal} {New Journal of Physics}\ }\textbf {\bibinfo {volume} {14}},\ \bibinfo {pages} {043020} (\bibinfo {year} {2012})}\BibitemShut {NoStop}%
\bibitem [{\citenamefont {Short}\ and\ \citenamefont {Farrelly}(2012)}]{ShortFarrelly11}%
  \BibitemOpen
  \bibfield  {author} {\bibinfo {author} {\bibfnamefont {A.~J.}\ \bibnamefont {Short}}\ and\ \bibinfo {author} {\bibfnamefont {T.~C.}\ \bibnamefont {Farrelly}},\ }\href {\doibase 10.1088/1367-2630/14/1/013063} {\bibfield  {journal} {\bibinfo  {journal} {New. J. Phys.}\ }\textbf {\bibinfo {volume} {14}},\ \bibinfo {pages} {013063} (\bibinfo {year} {2012})}\BibitemShut {NoStop}%
\bibitem [{\citenamefont {Bohigas}\ \emph {et~al.}(1995)\citenamefont {Bohigas}, \citenamefont {Giannoni}, \citenamefont {de~Almeida},\ and\ \citenamefont {Schmit}}]{bohigas1995chaotic}%
  \BibitemOpen
  \bibfield  {author} {\bibinfo {author} {\bibfnamefont {O.}~\bibnamefont {Bohigas}}, \bibinfo {author} {\bibfnamefont {M.-J.}\ \bibnamefont {Giannoni}}, \bibinfo {author} {\bibfnamefont {A.~O.}\ \bibnamefont {de~Almeida}}, \ and\ \bibinfo {author} {\bibfnamefont {C.}~\bibnamefont {Schmit}},\ }\href@noop {} {\bibfield  {journal} {\bibinfo  {journal} {Nonlinearity}\ }\textbf {\bibinfo {volume} {8}},\ \bibinfo {pages} {203} (\bibinfo {year} {1995})}\BibitemShut {NoStop}%
\bibitem [{\citenamefont {Deift}(2000)}]{deift2000orthogonal}%
  \BibitemOpen
  \bibfield  {author} {\bibinfo {author} {\bibfnamefont {P.}~\bibnamefont {Deift}},\ }\href@noop {} {\emph {\bibinfo {title} {Orthogonal Polynomials and Random Matrices: A Riemann-Hilbert Approach: A Riemann-Hilbert Approach}}},\ Vol.~\bibinfo {volume} {3}\ (\bibinfo  {publisher} {American Mathematical Soc.},\ \bibinfo {year} {2000})\BibitemShut {NoStop}%
\bibitem [{\citenamefont {Cipolloni}\ \emph {et~al.}(2023)\citenamefont {Cipolloni}, \citenamefont {Erd{\H{o}}s},\ and\ \citenamefont {Schr{\"o}der}}]{cipolloni2023spectral}%
  \BibitemOpen
  \bibfield  {author} {\bibinfo {author} {\bibfnamefont {G.}~\bibnamefont {Cipolloni}}, \bibinfo {author} {\bibfnamefont {L.}~\bibnamefont {Erd{\H{o}}s}}, \ and\ \bibinfo {author} {\bibfnamefont {D.}~\bibnamefont {Schr{\"o}der}},\ }\href@noop {} {\bibfield  {journal} {\bibinfo  {journal} {Communications in Mathematical Physics}\ ,\ \bibinfo {pages} {1}} (\bibinfo {year} {2023})}\BibitemShut {NoStop}%
\bibitem [{\citenamefont {Forrester}(2021{\natexlab{a}})}]{forrester2021differential}%
  \BibitemOpen
  \bibfield  {author} {\bibinfo {author} {\bibfnamefont {P.~J.}\ \bibnamefont {Forrester}},\ }\href@noop {} {\bibfield  {journal} {\bibinfo  {journal} {Journal of Statistical Physics}\ }\textbf {\bibinfo {volume} {183}},\ \bibinfo {pages} {33} (\bibinfo {year} {2021}{\natexlab{a}})}\BibitemShut {NoStop}%
\bibitem [{\citenamefont {Forrester}(2021{\natexlab{b}})}]{forrester2021quantifying}%
  \BibitemOpen
  \bibfield  {author} {\bibinfo {author} {\bibfnamefont {P.~J.}\ \bibnamefont {Forrester}},\ }\href@noop {} {\bibfield  {journal} {\bibinfo  {journal} {Communications in Mathematical Physics}\ }\textbf {\bibinfo {volume} {387}},\ \bibinfo {pages} {215} (\bibinfo {year} {2021}{\natexlab{b}})}\BibitemShut {NoStop}%
\bibitem [{\citenamefont {Liu}(2018)}]{liu2018spectral}%
  \BibitemOpen
  \bibfield  {author} {\bibinfo {author} {\bibfnamefont {J.}~\bibnamefont {Liu}},\ }\href@noop {} {\bibfield  {journal} {\bibinfo  {journal} {Physical Review D}\ }\textbf {\bibinfo {volume} {98}},\ \bibinfo {pages} {086026} (\bibinfo {year} {2018})}\BibitemShut {NoStop}%
\end{thebibliography}%

  \onecolumngrid

\appendix

\setcounter{lemma}{0}
\setcounter{theorem}{0}

\section{RMT predictions for $q=2,3$} \label{sec:RMTq23}

In this section we demonstrate that random matrix models have Poisson statistics for $q =  2,3$.  We accomplish this by simply looking at the ratio test outlined in the main text. First let us define a random matrix Hamiltonian, 

\begin{equation} \label{eq:RMTHAM}
    \hat{H} = A + A^T, 
\end{equation}
where $A$ is a matrix filled with random numbers generated from a normal distribution with zero mean and unit variance. We label the side length of $A$ as $N$. Similar to the physical Hamiltonian, we can study the $q = 2$ case by first constructing the Hamiltonian
\begin{equation} 
		    \hat{H}_2 = \hat{H} \otimes \mathbb{I} + \mathbb{I} \otimes \hat{H},
		\end{equation}

  with the spectrum $\Lambda_{k,l}$. Again this spectrum is symmetric under permutations of the indices, so we resolve this symmetry and only treat eigenvalues with unique $(k,l)$ such that $l > k$. We investigate the spectral properties of this Hamiltonian after resolving our symmetry in figure \ref{fig:RMTq23} in the left panel. Here we clearly see agreement with a Poisson distribution. The random matrix model experiences Poisson statistics.

    \begin{figure}[h!]
    \centering
    \includegraphics[width=0.45\linewidth]{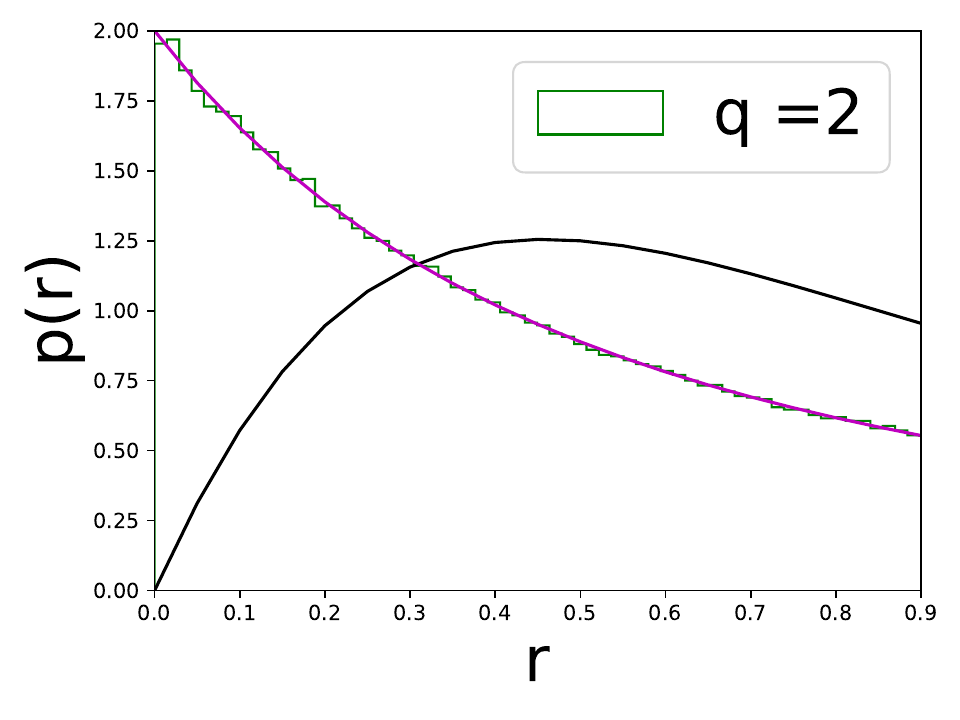}
\includegraphics[width=0.45\linewidth]{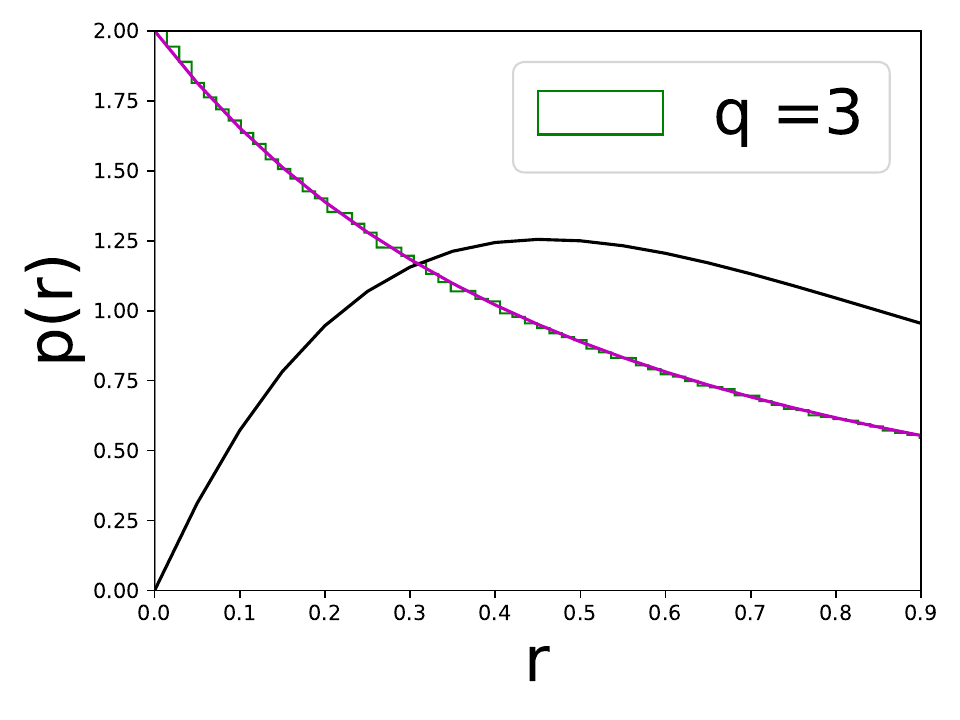}
\caption{ (a) Ratio test for the $q = 2 $ symmetry resolved random matrix Hamiltonian. Side length is $N = 1200$. (b) Ratio test for the $q = 3$ symmetry resolved random matrix Hamiltonian. Side length is $N = 200$.     }
\label{fig:RMTq23}
\end{figure}
The construction of a new Hamiltonian for $q = 3$ is similar. We have 

\begin{equation}
    \hat{H}_3 = \hat{H} \otimes \mathbb{I} \otimes \mathbb{I}   + \mathbb{I} \otimes \hat{H} \otimes \mathbb{I}  + \mathbb{I}  \otimes \mathbb{I}  \otimes \hat{H}.
\end{equation}
This gives us a new spectrum of $\Lambda_{k,l,q} = E_k + E_l + E_q$, where this new spectrum is also invariant under permutations of its indices. We resolve this symmetry by considering terms such that $q>l>k$. The result of the ratio test on this new spectrum is given in the right panel of Fig. \ref{fig:RMTq23}, indicating again Poisson statistics. This is similarly found for higher values of $q$, which leads us to conjecture that this will be true for all $q\geq2$.

\section{Physical Hamiltonian q = 3,4 spectral statistics}\label{sec:q34}

In this section we provide numerical evidence of Poisson statistics for $q = 3,4$ in the physical Hamiltonian. We repeat the $q=3$ case as was covered in the RMT appendix, and also investigate the $q=4$ statistics. Both cases will be covered with the ratio test, and we will use the same physical model as the main text, where we resolve all relevant symmetries. For the $q = 4$ case, we must work with the Hamiltonian, 

\begin{equation}
    \hat{H}_4 = \hat{H} \otimes \mathbb{I} \otimes \mathbb{I} \otimes \mathbb{I}    + \mathbb{I} \otimes \hat{H} \otimes \mathbb{I} \otimes \mathbb{I}   + \mathbb{I} \otimes \mathbb{I}  \otimes \hat{H} \otimes \mathbb{I}  + \mathbb{I} \otimes \mathbb{I} \otimes \mathbb{I}  \otimes\hat{H}.
\end{equation}

This gives us a new spectrum, again which is invariant under index permutations. We can resolve this symmetry with an identical strategy to the $q=2,3$ cases, and study the corresponding symmetry-resolved spectrum. The results of this for $q=3,4$ in the physical Hamiltonian are given in Fig. \ref{fig:q34}. These results again indicate that the spectrum  obeys Poisson statistics. The left panel in Fig. \ref{fig:q34} also serves as evidence that the statistics of the Hamiltonian agrees with RMT as seen in the right panel of Fig. \ref{fig:RMTq23}.

   \begin{figure}[h!]
    \centering
    \includegraphics[width=0.45\linewidth]{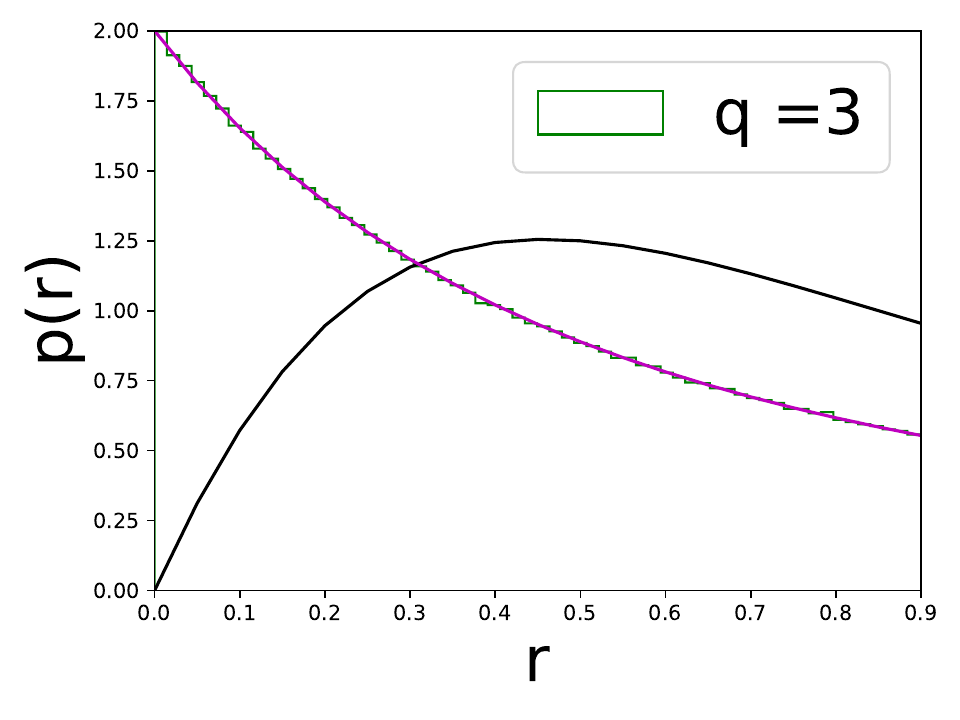}
\includegraphics[width=0.45\linewidth]{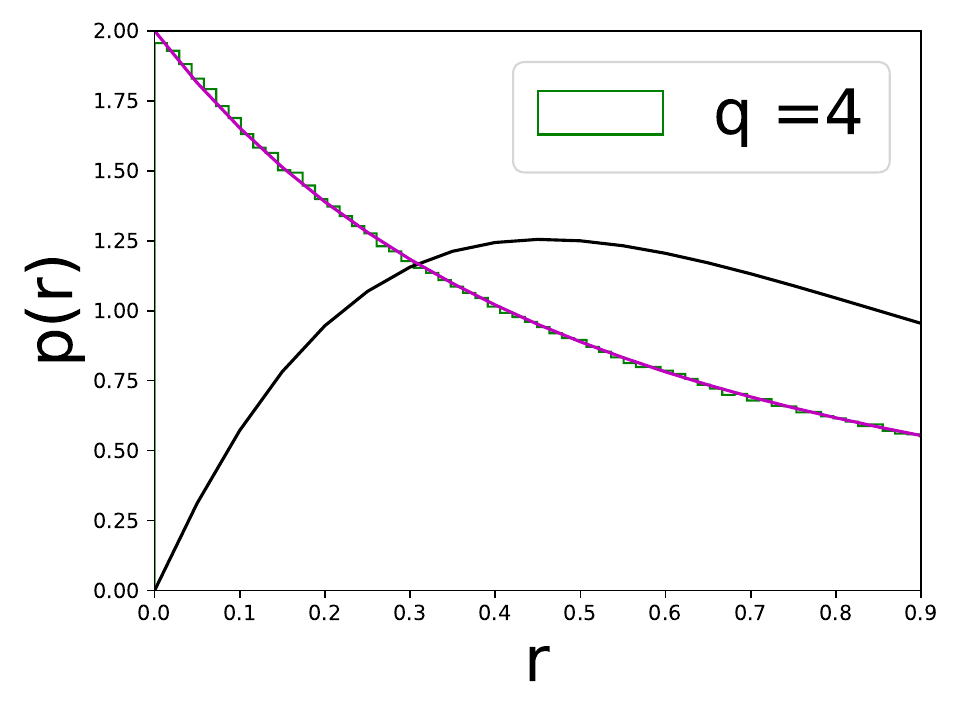}
\caption{ (a) Ratio test for the $q = 3 $ symmetry resolved physical Hamiltonian. The physical Hamiltonian $\hat{H}$ is generated with $L = 16$. (b) Ratio test for the $q = 4$ symmetry resolved physical Hamiltonian. The Hamiltonian $\hat{H}$ here has $L = 14$.     }
\label{fig:q34}
\end{figure}

\end{document}